\documentclass[runningheads]{llncs}

\usepackage[T1]{fontenc}
\usepackage{graphicx}
\usepackage{amsmath}
\usepackage{amssymb}
\usepackage{color}
\usepackage{enumitem}
\usepackage[breaklinks,colorlinks=true, citecolor=blue]{hyperref}
\usepackage{dsfont}
\usepackage{mathrsfs}
\usepackage{mathtools}
\usepackage[ruled,vlined,linesnumbered]{algorithm2e}
\DontPrintSemicolon
\usepackage{cite}

\usepackage[framemethod=tikz]{mdframed}
\usepackage{multicol}
\usepackage{booktabs}
\usepackage{siunitx}
\usepackage{multirow}
\usepackage{listings}
\usepackage{xcolor}

\renewcommand{\vec}[1]{\mathbf{#1}}

\newcommand{\threebinom}[2]{\ensuremath{{\genfrac{[}{]}{0pt}{}{#1}{#2}}_3}\xspace}

\newcommand{\Gen}{\mathsf{Gen}}
\newcommand{\Sign}{\mathsf{Sign}}
\newcommand{\Vrfy}{\mathsf{Verify}}
\newcommand{\advA}{\mathcal{A}}

\newcommand{\F}{\mathbb{F}}
\newcommand{\FF}{\F}

\newcommand{\Z}{\mathbb{Z}}
\newcommand{\ZZ}{\mathbb{Z}}
\newcommand{\hash}{\ensuremath{\mathsf{Hash}}\xspace}
\newcommand{\weight}{\mathsf{Weight}}

\newcommand{\htopoint}{\mathsf{HashToPoint}}

\newcommand{\decomp}{\mathsf{Decompress}}

\newcommand{\vcheck}{\ensuremath{\vec{v}_{\text{check}}}\xspace}
\newcommand{\vchecki}[1][i]{\ensuremath{\vec{v}_{\text{check},{#1}}}\xspace}

\newcommand{\Constraint}[1]{\ensuremath{\textsc{Constraint}(#1)}}
\newcommand{\Ring}{\ensuremath{\mathcal{R}}\xspace}
\newcommand{\Module}{\ensuremath{\mathcal{M}}\xspace}

\newcommand{\RedModule}{\ensuremath{\overline{\Module}}\xspace}
\newcommand{\Kernel}{\ensuremath{\mathcal{K}}\xspace}
\newcommand{\KernelSet}{\ensuremath{\mathcal{S}}\xspace}
\newcommand{\TSet}{\ensuremath{\mathcal{T}}\xspace}

\newcommand{\SGP}{\ensuremath{\mathsf{SEGP}}\xspace}
\newcommand{\Sigmacomp}[1][{}]{\ensuremath{\Sigma_{\text{comp}}^{#1}}\xspace}

\newcommand{\Lattice}{\ensuremath{\mathcal{L}}\xspace}

\newcommand{\kmax}{\ensuremath{k_\mathrm{max}}\xspace}
\newcommand{\kmin}{\ensuremath{k_\mathrm{min}}\xspace}

\newcommand{\Wave}[1][{}]{{\textsc{Wave{#1}}}\xspace}
\newcommand{\Squirrels}[1][{}]{{\textsc{Squirrels{#1}}}\xspace}
\newcommand{\squirrels}{\Squirrels}

\newcommand{\SPHINCS}[1][{}]{{\textsc{SPHINCS+{#1}}}\xspace}
\newcommand{\XMSSMT}[1][{}]{{\textsc{XMSS{${}^{MT}$}{#1}}}\xspace}

\newcommand{\secretprime}{\ensuremath{r}\xspace}
\newcommand{\secretprimecount}{\ensuremath{t}\xspace}
\newcommand{\secretprimes}{\ensuremath{\vec{r}}\xspace}
\newcommand{\squirrelsprime}{\ensuremath{p}\xspace}
\newcommand{\squirrelsprimecount}{\ensuremath{s}\xspace}
\newcommand{\squirrelsprimes}{\ensuremath{\vec{p}}\xspace}

\newcommand{\RNS}[2]{\ensuremath{[[#1]]_{#2}}\xspace}

\SetKwInOut{Parameters}{Parameters}
\SetKwProg{Function}{function}{}{}
\SetKwFunction{qCoefficients}{qCoefficients}
\SetKwFunction{ModECRTSetup}{ModECRTSetup}
\SetKwFunction{ModECRT}{ModECRT}
\SetKwFunction{InverseMod}{InverseMod}
\SetKwFunction{RandomPrime}{RandomPrime}
\SetKwFunction{PKeyGen}{KeyGen}
\SetKwFunction{Sign}{Sign}
\SetKwFunction{Verify}{Verify}
\SetKwFunction{CKeyGen}{CKeyGen}
\SetKwFunction{VKeyGen}{VKeyGen}
\SetKwFunction{CVerify}{CVerify}
\SetKwData{salt}{salt}
\SetKwData{CK}{CK}
\SetKwData{PK}{PK}
\SetKwData{SK}{SK}
\SetKwData{VK}{VK}
\SetKwData{Reject}{Reject}
\SetKwData{Accept}{Accept}
\SetKwData{True}{True}
\SetKwData{False}{False}

\newcommand{\pk}{\PK}
\newcommand{\vk}{\VK}

\newcommand{\Primes}[1]{\ensuremath{\textsc{Primes}({#1})}\xspace}

\usepackage{xargs}
\usepackage[colorinlistoftodos,prependcaption,textsize=tiny]{todonotes}
\newcommandx{\fixme}[2][1=]{}
\newcommandx{\change}[2][1=]{}

\newcommandx{\improvement}[2][1=]{}

\newcommandx{\diff}[1]{\textcolor{blue}{#1}\xspace}

\newtheorem{assumption}{Assumption}

\usepackage[
n,advantage,
operators,
sets,
adversary,
landau,
probability,
notions,
logic,
ff,
mm,
primitives,
events,
complexity,
oracles,
asymptotics,
keys
]{cryptocode}

\begin{document}
\title{Compressed verification for post-quantum signatures with long-term public keys}
\titlerunning{Compressed public keys for conservative post-quantum signatures}

\author{Gustavo Banegas
\and Ana\"elle Le Dévéhat
\and Benjamin Smith}

\authorrunning{Banegas, Le Dévéhat, Smith}
\institute{
Inria and Laboratoire d’Informatique de l’Ecole polytechnique,\\
  Institut Polytechnique de Paris,
  Palaiseau, France\\
  \email{gustavo@cryptme.in}\\
  \email{anaelle.le-devehat@inria.fr}\\
  \email{smith@lix.polytechnique.fr}
}
\maketitle              \vspace{-1ex}
  \makeatletter \begingroup \makeatletter \def\@thefnmark{$*$}\relax \@footnotetext{\relax Author list in alphabetical order; see \url{https://ams.org/profession/leaders/CultureStatement04.pdf}. 
    This work was supported by the HYPERFORM consortium, funded by France through Bpifrance, and by the France 
    2030 program under grant agreement ANR-22-PETQ-0008 PQ-TLS.
\def\ymdtoday{\leavevmode\hbox{\the\year-\twodigits\month-\twodigits\day}}\def\twodigits#1{\ifnum#1<10 0\fi\the#1}Date of this document: \ymdtoday.}\endgroup
\begin{abstract}

    Many signature applications---such as root certificates, 
    secure software updates, and authentication protocols---involve 
    long-lived public keys that are transferred or installed once 
    and then used for many verifications.
    This key longevity makes post-quantum signature schemes with 
    conservative assumptions (e.g., structure-free lattices) 
    attractive for long-term security.
    But many such schemes, especially those with short 
    signatures, suffer from extremely large public keys. Even 
    in scenarios where bandwidth is not a major concern, large 
    keys increase storage costs and slow down verification.
    We address this with a method to replace large public keys in 
    GPV-style signatures with smaller, private verification keys. 
    This significantly reduces verifier storage and 
    runtime while preserving security. Applied to
    the conservative, short-signature schemes 
    \Wave and \Squirrels,
    our method compresses \Squirrels[-I] keys from 
    \SI{665}{\kilo\byte} to \SI{20.7}{\kilo\byte} and \Wave[822] keys 
     from \SI{3.5}{\mega\byte} to \SI{207.97}{\kilo\byte}.

\end{abstract}
 
\section{Introduction
}

Post-quantum signatures 
are a primary requirement for the transition towards quantum-resistant
cryptography.
Post-quantum lattice- and code-based signatures
can be roughly classified as \emph{conservative}
or \emph{structured},
according to whether their underlying hard problems
involve general codes and lattices,
or involve special algebraic structure.
These structures
facilitate important practical improvements
(often, much smaller public keys);
but they also allow the possibility of specialized attacks,
so structured schemes generally have weaker security arguments
(i.e., their assumptions are stronger). \change{maybe stress this even more to try and avoid comments saying squirrels and wave are not in the nist etc...}

For example: 
compare the structured lattice scheme \textsc{Falcon}~\cite{FALCON}
with the conservative lattice scheme \Squirrels~\cite{squirrels_sub}.
Both are based on the same GPV design~\cite{GPV08}.
At NIST post-quantum security level~1,
\textsc{Falcon} and \Squirrels
have comparable signature sizes:
666 and 1019 bytes, respectively.
But while the structured lattices of \textsc{Falcon}
give 897-byte public keys,
the unstructured lattices of \Squirrels push public-key sizes
up to 665 \emph{kilo}bytes.

Signature schemes with large public keys are unsuitable 
for applications where public keys are regularly transmitted,
such as TLS certificates.
They are better-suited to applications where
\begin{itemize}
    \item
        the public key is pre-installed on the verifier's device
        (for verifying signed software updates, for example,
        or root certificates), or
    \item
        the public key is transmitted, but the cost of transmission is
        amortised over many subsequent verifications
        (in \texttt{ssh} authentication, for example).
\end{itemize}
These applications often involve public keys with \emph{very} long 
lifetimes: 20-30 years for root certificates like ISRG Root X1 and 
GlobalSign Root R1, for example,
and a decade or more for IoT code-signing certificates
and government-issued digital IDs.

The long-term nature of keys in these applications makes conservative 
security assumptions reassuring, but working with 
very large public keys remains expensive and inconvenient.
This makes hash-based signatures like SLH-DSA
(SPHINCS+)~\cite{SPHINCSplus,SLH-DSA} an interesting choice:
they offer conservative security assumptions \emph{and}
very small public keys---but at the cost of large signatures and computationally 
intensive verification. \Squirrels
offers shorter signatures and faster verification, 
but its \SI{665}{\kilo\byte} public keys make long-term storage 
impractical.

\subsection{Compressed verification}
\label{subsec:Cverif}

We want to reduce public key storage for conservative GPV-style
signatures.
The core idea is that the verifier can (pre-)process the public key \PK once
to derive a much smaller verification key \VK (private to the verifier),
which they can then use in place of \PK
for confident---and often much faster---signature verification.

More formally: suppose we are given a signature scheme defined by
three algorithms,
with an implicit security parameter \(\lambda\):
\begin{itemize}
    \item
        \PKeyGen:
        returns a private key \SK
        and a public key \PK.
    \item
        \Sign:
        given a private key \SK
        and a message \(m\),
        returns a signature \(\sigma\). \item
        \Verify:
        given a putative signature \(\sigma\)
        on \(m\) under a public key \PK,
        returns \Accept or \Reject.
\end{itemize}
We will define three additional algorithms
to be used by the verifier:
\begin{itemize}
    \setcounter{enumi}{3}
    \item
        \CKeyGen:
        returns a (private) compression key \CK.
    \item
        \VKeyGen:
        given a public key \PK
        and a (private) compression key \CK,
        return a private verification key \VK.
    \item
        \CVerify:
        given a putative signature \(\sigma\) on \(m\)
        and a verification key \VK,
        returns \Accept or \Reject.
\end{itemize}
The goal is to define these functions such that
if \(\VK = \VKeyGen(\CK,\PK)\)
for some public key \PK and some \CK output by \CKeyGen,
then
\begin{enumerate}
    \item
        if \Verify{\(\sigma\), \(m\), \PK} = \Accept
        then \CVerify{\(\sigma\), \(m\), \VK}
        = \Accept;
    \item
        if \Verify{\(\sigma\), \(m\), \PK} = \Reject
        then \CVerify{\(\sigma\), \(m\), \VK} 
        = \Reject
        with probability $\ge 1 - 1/2^\mu$
        for a second security parameter \(\mu\);
        and
    \item
        the size of \VK
        is much smaller than the size of \PK.
\end{enumerate}
In terms of storage, $\CK$ is generated randomly 
and---once it has been used in \VKeyGen---need not be stored.
Therefore, the verifier only needs to retain \VK.

The verifier is
free to choose \(\mu\);
our goal is that
the probability that \CVerify accepts one forgery
after \(Q\) attempts is on the order of \(Q/\#\KernelSet\),
where \(\KernelSet\) is the verifier's compression-keyspace.
We will generally take \(\#\KernelSet \approx 2^{\mu}\)
with \(\mu \approx \lambda\),
assuming \(Q\) is relatively small
(in any case, \(Q \le 2^{64}\)).
The verifier may force a limit on \(Q\)
by refreshing \(\VK\)
after a given number of rejections.

Figure~\ref{fig:compressed-protocol} illustrates the 
signature protocol with compressed verification.  
\CK and \VK are private \emph{to the verifier} and are 
derived only from the signer's public key \PK,
and not the signer's private key \SK.
Additionally, the verification 
algorithm \CVerify does not require \PK or \CK.
    From the signer's point of view, the original signature scheme is unchanged.

\begin{figure}
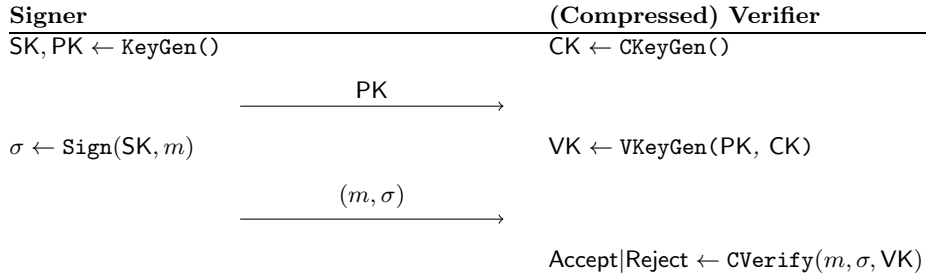

\centering
\pseudocodeblock[codesize=\footnotesize, bodylinesep=0.1\baselineskip]{
    \textbf{Signer} \<\< \textbf{(Compressed) Verifier}
    \\[][\hline]
    \SK, \PK \gets \PKeyGen{}\<\< \CK \gets \CKeyGen{} 
    \\
    \< \sendmessageright{top=\PK} \<
    \\
    \sigma \gets \Sign(\SK, m) \<\< \VK \gets \VKeyGen{\PK, \CK} 
    \\
    \< \sendmessageright*{(m,\sigma)} \<
    \\
    \<\<  \Accept \vert \Reject \gets \CVerify(m, \sigma, \VK) 
}
\caption{Compressed verification as a protocol.}
    \label{fig:compressed-protocol}
\end{figure}

\subsection{Results}
We introduce a framework for \emph{compressed verification} in
GPV-style signatures. 
The verifier chooses a private homomorphism \(\phi\)
and uses it to compress incoming public keys to compact 
\emph{private verification keys}.
We give security arguments for the general 
construction, and consider two conservative instantiations: 
\Squirrels~\cite{squirrels_sub}, a lattice-based signature, 
and \Wave~\cite{Wave,Wave-submission}, a code-based signature. 

\Squirrels and \Wave 
both have little special algebraic structure. 
This builds confidence in their security,
but it comes at an important practical cost:
as we see in Table~\ref{tab:results},
their public keys are \emph{very} big.  
Indeed, among the submissions to the NIST Postquantum Signatures
on-ramp with no vulnerabilities found during Round~1,
\Squirrels and \Wave had the largest public keys at most security
levels.\footnote{The \emph{classic} parameters of
    \href{https://www.uovsig.org/}{UOV},
    a conservative multivariate scheme,
    had larger public keys than \Squirrels at NIST PQ Security Level 5,
    but other UOV variants had smaller keys.
    In any case, \Wave is the undisputed super-heavyweight champion:
    its Level I public keys were larger than even the Level V keys of any other
    scheme.}
These oversized public keys would be a major factor
in the non-selection of these schemes for Round~2
of the standardization process.
They also make \Squirrels and \Wave prime candidates
to demonstrate the effectiveness of our technique.

We developed two implementations 
of both \Squirrels and \Wave:
\begin{enumerate}
    \item
        A comprehensive implementation in Python,
        and
\item
        A C implementation,
        based on the reference implementations
        of \squirrels~\cite{squirrels_impl}
and \Wave,
        used to measure performance improvements.
\end{enumerate}

Table~\ref{tab:results} shows the significant 
reduction in verifier storage achieved for 
both \Squirrels and \Wave while maintaining
security levels (and backward compatibility with the original schemes). \fixme{give a more detailled sentence on security here, add an assumption ?}
Specifically, we achieve a \textbf{compression ratio of up to 34x} for 
\Wave at higher security levels, and \textbf{over 26x} for \Squirrels at
comparable levels.
Compressed verification can also reduce verification time:
verification is \textbf{up to 9.26\% faster}
than ordinary verification for 
\Squirrels and \textbf{up to 30\% faster} for \Wave.

\begin{table}
    \caption{Signature scheme parameters and verification times.
        All sizes are in bytes.
        Compressed verification parameters are chosen such that
        \(\mu\approx\lambda\) 
        (see~\S\ref{sec:sq_mar} and~\S\ref{sec:sketch_wave_rip} for details).
        Verification times are cycle counts
        on an Intel Core i7-1365U processor
        running Arch Linux Kernel 6.11.5-arch1.
        ``Reference'' refers to C reference implementations,
        and ``Compressed'' to our own C code.
        \Wave signatures are variable-length;
        sizes here are upper bounds,
        and \Wave signature length roughly doubles for compressed
        verification.
    }
    \label{tab:results}
    \centering
    \resizebox{\textwidth}{!}{
        \begin{tabular}{l|@{\ }rr|rrr|rrr}
            \toprule
            & \multicolumn{2}{c|}{Reference}
            & \multicolumn{3}{c|}{Compressed verification}
            & \multicolumn{3}{c}{Verification time (kCycles)}
            \\
            & \(|\sigma|\) & \(|\PK|\)
            & \(|\VK|\)
            & (\(|\CK|\))
            & \(|\PK|/|\VK|\)
            & {Reference}
            & {Compressed}
            & \emph{Speedup}
            \\
            \midrule
            \multicolumn{9}{l}{\emph{NIST PQ Security Level 1, classical \(\lambda = 128\)}}
            \\
            \Squirrels[-I]
            &  1\,019 &     \textbf{681\,780}
            & \textbf{20\,700}
            &  (3\,360) 
            & \emph{32.9\(\times\)}
            &     280
            &     254
            & \emph{9.3\%}
            \\
            \Wave[822]
            &     822 &  \textbf{3\,677\,390}
            & \textbf{171\,594}
            &   (83\,822)
            & \emph{21.4\(\times\)}
            &  1\,101
            &     762
            & \emph{30.8\%}
            \\
            {\SPHINCS[-SHAKE-128s]}
            &  7\,856 & 32
            & ---
            & ---
            & ---
            &  3\,285
            & ---
            & ---
            \\
            {\XMSSMT-20-2}
            &  4\,963 & 64
            & ---
            & ---
            & ---
            & 2\,868
            & ---
            & ---
            \\
            \midrule
            \multicolumn{9}{l}{\emph{NIST PQ Security Level 3, classical \(\lambda = 192\)}}
            \\
            \Squirrels[-III]
            &  1\,554 &  \textbf{1\,629\,640}
            & \textbf{49\,824}
            &  (8\,480)
            & \emph{32.7\(\times\)}
            &     551
            &     520
            & \emph{5.69\%}
            \\
            \Wave[1249]
            &  1\,249 &  \textbf{7\,867\,598}
            & \textbf{266\,188}
            & (183\,134)
            & \emph{29.6\(\times\)}
            & 2\,330
            & 1\,865
            & \emph{19.9\%}
            \\
            {\SPHINCS[-SHAKE-192s]}
            & 16\,224 & 48
            & ---
            & ---
            & ---
            &  5\,374
            & ---
            & ---
            \\
            {\XMSSMT-40-4}
            &  9\,893 & 64
            & ---
            & ---
            & ---
            &  5\,729
            & ---
            & ---
            \\
            \midrule
            \multicolumn{9}{l}{\emph{NIST PQ Security Level 5, classical \(\lambda = 256\)}}
            \\
            \Squirrels[-V]
            &  2\,025 &  \textbf{2\,786\,580}
            &  \textbf{90\,598}
            & (15\,048)
            & \emph{30.8\(\times\)}
            &     916
            &     898
            & \emph{1.9\%}
            \\
            \Wave[1644]
            &  1\,644 & \textbf{13\,632\,308}
            & \textbf{370\,436}
            & (321\,507)
            & \emph{36.8\(\times\)}
            &  3\,911
            &  3\,198
            & \emph{18.2\%}
            \\
            {\SPHINCS[-SHAKE-256s]}
            & 29\,792 & 64
            & ---
            & ---
            & ---
            &  7\,567
            & ---
            & ---
            \\
            {\XMSSMT-60-6}
            & 14\,824 & 64
            & ---
            & ---
            & ---
            &  8\,743
            & ---
            & ---
            \\
            \bottomrule
        \end{tabular}
    }
\end{table}

As a baseline for conservative signatures,
Table~\ref{tab:results}
includes the hash-based schemes
\SPHINCS and \XMSSMT~\cite{XMSSMT}.
For \SPHINCS, the authors provide a benchmark script,
which we ran in our environment; in the table 
we report the fastest result produced by that
script~\cite{SPHINCSPLUS_impl}.
For \XMSSMT, we used the reference 
code~\cite{XMSS_impl}
with the smallest parameters for each security level.

\fixme{The ``illustrative example'' from the conclusion would be better
here.  It's currently commented out here; if we have space, we can add
it back in or something like it.}

\begin{remark}
Compressed verification may also benefit schemes like Ajtai-based 
    hash-and-sign~\cite{2012/MP,2019/CGN}, which offer strong SIS-based security.
    The benefit is limited for
    structured GPV-style schemes such as \textsc{Falcon}~\cite{FALCON}:
    public keys are already compact, and compressed verification requires the full \((s_1, s_2)\) 
    signature rather than just $s_2$, thus doubling signature size.
\end{remark}

\begin{remark}
    Our method is compatible with the PS-3~\cite{2005/Pornin--Stern} 
    and BUFF~\cite{cremers2021buffing} transforms, which strengthen 
    security in certain attack models. However, both require hashing
    \PK with the message—a costly step when \PK is large. 
    We suggest storing the much smaller $\hash(\PK)$ and 
    hashing that with the message instead.
\end{remark}

\subsection{Related work: flexible signatures and progressive verification}
\label{sec:flexible-progressive}

Fischlin's \emph{progressive verification} for MACs~\cite{2003/Fischlin}
can probabilistically \Reject early or \Accept with reduced confidence;
\cite[\S4]{2003/Fischlin} suggests an extension to signatures.
Le, Kelkar, and Kate's \emph{flexible
signatures}~\cite{2019/Le--Kelkar--Kate},
verified up to a real ``confidence level'' \(0 \le \alpha \le 1\),
quantify ``partial'' verification when expensive verification
operations are interrupted voluntarily by the user
or forcibly by the OS.
This can improve fault-tolerance and reduce the cost of
verification in embedded applications,
but it does not reduce public-key or signature sizes.
Indeed, the main targets in~\cite{2019/Le--Kelkar--Kate}
are hash-based signatures, where public keys
are already extremely compact;
but extensions to GPV signatures
are projected in~\cite[\S5.3]{2019/Le--Kelkar--Kate}.

Taleb and Vergnaud revisit progressive verification 
in~\cite{2020/Taleb--Vergnaud}, analyzing Bernstein's 
RSA trick (see~\S\ref{sec:rabin})
and GPV signatures (including \Wave). They propose 
verifying GPV signatures using a small set of linear combinations 
of columns based on a random linear code, achieving exponential 
confidence growth with runtime, unlike the linear growth 
in~\cite{2003/Fischlin} and~\cite{2019/Le--Kelkar--Kate}. 
However, their approach significantly increases public key size, 
which is the opposite of our goal.

Boschini, Fiore, Pagnin, Torresetti, and Visconti~\cite{BoschiniFPTV24}
propose an \emph{efficient verification} for signatures
which verify using a matrix-vector product $\vec{M}\vec{v}^\top$.
In an ``offline'' phase they compute a matrix $\vec{M}'$
formed by $k$ random linear combinations of the $n$ rows of $\vec{M}$;
then, in an ``online'' phase, they verify using $\vec{M}'\vec{v}^\top$,
with reduced confidence but with a speedup of $n/k$.
This is very similar to what we do with \Wave
in~\S\ref{sec:sketch_wave_rip},
but they repeat the offline phase for every verification
rather than maintaining the same $\vec{M}'$;
indeed, their focus is minimising online verification latency,
rather than reducing overall verification time or key sizes.

\section{Warmup: Bernstein's trick for Rabin--Williams
}\label{sec:rabin}

As a warmup,
we recall Bernstein's fast Rabin--Williams signature verification~\cite{2008/Bernstein}.
We write \(\Primes{\mu}\) for the set of (exactly) \(\mu\)-bit primes:
that is, 
\[
    \Primes{\mu} 
    := 
    \{ 2^{\mu-1} < p < 2^{\mu} \mid p \text{ is prime} \}
    \,. 
\]

\subsection{Verifying Rabin--Williams signatures}

A Rabin--Williams signature~\cite{1979/Rabin,1980/Williams}
on a message \(m\)
under a public key \(N = pq\)
is a tuple \(\sigma = (e, f, \salt, s)\) such that
\begin{equation}
    \label{eq:RW-verification}
    efs^2 \equiv \hash(\salt\parallel m)\pmod{N}
\end{equation}
where
\(1 < s < N\),
\(\salt\) is a salt value, and 
\(e\in\{-1,1\}\) and \(f \in \{1,2\}\)
are chosen such that \(s\) exists
for the given \(\salt\), \(m\), and \(N\).
That is:
given \(\sigma = (e,f,\salt,s)\), \(m\), and \(\PK = N\),
\Verify{\(\sigma\), \(m\), \PK}
returns \Accept
if and only if
\eqref{eq:RW-verification} holds.

If~\eqref{eq:RW-verification} is satisfied,
then there is a unique integer \(-2N < t < 2N\)
such that
\begin{equation}
    \label{eq:Bernstein-RW-expanded}
    efs^2 - tN = \hash(\salt\parallel m)
    \,.
\end{equation}
(The sign of \(t\) is equal to \(e\).)
Note that~\eqref{eq:Bernstein-RW-expanded}
holds over \(\Z\),
not just over \(\Z/N\Z\);
and any solution to~\eqref{eq:Bernstein-RW-expanded}
yields a solution to~\eqref{eq:RW-verification}
and vice versa,
so verifying~\eqref{eq:RW-verification}
or~\eqref{eq:Bernstein-RW-expanded}
is mathematically (though not algorithmically) equivalent.

Bernstein suggested speeding up verification by including \(t\) in 
\(\sigma\) and verifying~\eqref{eq:Bernstein-RW-expanded} modulo a 
random \(\lambda\)-bit prime \(\ell\) (with \(\lambda\) the security 
parameter). 
The verifier picks \(\ell \in \Primes{\lambda}\), computes 
\(N_\ell := N \bmod \ell\), and upon receiving 
\(\sigma = (e, f, \salt, s, t)\), checks  
\begin{equation}
    \label{eq:Bernstein-RW-reduced}
    efs_\ell^2 - t_\ell N_\ell \equiv h_\ell \pmod \ell,
\end{equation}
where \(s_\ell := s \bmod \ell\), \(t_\ell := t \bmod \ell\), and 
\(h_\ell := \hash(\salt \parallel m) \bmod \ell\).  
Since \(\ell \ll N\), this is faster than computing 
\(s^2 \bmod N\); the speedup increases with~\(\lambda\).
The trade-off: including \(t\) doubles the size of \(\sigma\), 
and generating \(\ell\) is relatively costly. However, 
as Bernstein notes, \(\ell\) can be reused across multiple 
verifications---even for different public keys---if kept secret, 
amortizing the cost.

This trick was proposed to speed up verification.
We observe that it also saves space
if many signatures are verified under the same \(N\),
since \(N_\ell\) can be stored instead of \(N\)
(and this saves even more time,
since \(N_\ell\) need not be recomputed).

In terms of our framework above,
\begin{itemize}
    \item
        \CKeyGen samples a random \(\lambda\)-bit prime \(\ell\);
    \item
        \VKeyGen{\(\CK = \ell\), \(\PK = N\)}
        returns \(\VK = (\ell, N_\ell := N \bmod{\ell})\);
    \item
        \CVerify{\(m\), \(\sigma = (e,f,\salt,s,t)\), \(\VK = (\ell, N_\ell)\)}
        returns \Accept
        if and only if
        \( efs^2 - t N_\ell \equiv \hash(\salt\parallel m) \pmod{\ell}\).
\end{itemize}

As Bernstein observes, the same technique applies to 
RSA signatures. However, if \(e\) is the public exponent, 
the resulting integer \(t\) is roughly \(N^{e-1}\), 
making the signature \(e\) times longer than a standard RSA signature.

\subsection{Security argument for Bernstein's trick}
\label{sec:RW-security}

Let's say that an \(\ell\)-forgery on a message \(m\)
for a public key \(\PK = N\)
is a vector \((e,f,\salt,s,t)\)
such that~\eqref{eq:Bernstein-RW-expanded}
(and hence~\eqref{eq:RW-verification}) fails,
but~\eqref{eq:Bernstein-RW-reduced} holds.
That is: an \(\ell\)-forgery is a putative expanded Rabin--Williams signature
that \Verify with \(\PK = N\) would safely reject,
but \CVerify with \(\VK = (\ell, N \bmod{\ell})\)
would accept.
(We assume that forging a signature
for~\eqref{eq:RW-verification} or~\eqref{eq:Bernstein-RW-expanded}
is infeasible.)

An adversary that knows \(\ell\) can easily construct
\(\ell\)-forgeries for any \(m\) and~\(N\).
Take a random \salt,
and
find a \(t\) such that \(x := \hash(\salt\parallel m) + tN\)
is a square modulo \(\ell\);
then,
compute \(s := x^{1/2} \pmod{\ell}\) 
(which is easy because \(\ell\) is prime);
finally, set \(e := 1\) and \(f := 1\).

Conversely,
if we can find an \(\ell\)-forgery \(\sigma = (e,f,\salt,s,t)\)
for \((m,N)\)
then
\[
    \ell \mid \Xi(\sigma,m,N) 
    \qquad
    \text{where}
    \qquad 
    \Xi(\sigma,m,N) := efs^2-tN-\hash(\salt\parallel m)
    \,.
\]
At just \(\lambda\) bits,
the prime \(\ell\) is sufficiently small 
(for cryptographic values of \(\lambda\))
to be recovered from \(\Xi(\sigma,m,N)\)
with
ECM~\cite{1987/Lenstra,1987/Montgomery,2006/Zimmermann--Dodson}.\footnote{Further: if we can find two \(\ell\)-forgeries \(\sigma_1\) and \(\sigma_2\)---\emph{not} 
    necessarily for the same \(m\) and \(N\)---then
    \(\ell\mid g := \gcd(\Xi(\sigma_1,m_1,N_1),\Xi(\sigma_2,m_2,N_2))\),
    and in fact probably \(\ell = g\).
}

An adversary \(\mathcal{A}\) who can compute an
\(\ell\)-forgery can therefore find \(\ell\), and vice versa.
But \(\mathcal{A}\)'s interaction with the verifier
is limited to submitting tuples \((\sigma,m,N)\),
and observing whether they are accepted or rejected:
that is, whether 
the unknown \(\ell\) divides \(\Xi(\sigma,m,N)\) or not.
Let 
\[
    \mu := \log_2\ell
    \qquad
    \text{and}
    \qquad
    \kappa := \lfloor{(\log_2N)/\mu}\rfloor
    \,.
\]
Observe that \(\mathcal{A}\) learns nothing from an \(\ell\)-forgery
attempt \((\sigma,m,N)\) such that \(\Xi(\sigma,m,N)\)
is not divisible by at least one \(\mu\)-bit prime,
and that if \(\Xi(\sigma,m,N) \not= 0\),
then it is divisible by at most \(2\kappa\) \(\mu\)-bit primes
(because \(|\Xi(\sigma,m,N)| < 2N^2\)).
Therefore,
if an adversary makes at most \(Q\) forgery attempts
against a verifier using a \(\mu\)-bit prime \(\ell\) for \VK,
then their success probability is at most
\begin{equation}
    \label{eq:Rabin-kappa}
    P(N,\mu,Q)
    :=
    \frac{2\kappa Q}{\#\Primes{\mu}}
    \,.
\end{equation}
It follows from~\cite[Corollary~5.3]{Dusart18}
that
\begin{equation}
    \label{eq:Dusart-2}
    0.975\frac{2^{\mu-1}}{(\mu-1)\log2}
    <
    \#\Primes{\mu}
    <
    \frac{2^{\mu-1}}{(\mu-1)\log 2}.
\end{equation}
Thus, \(P(N, \mu, Q) \approx Q \cdot (\log N)/2^{\mu - 2}\). For 
\(\mu \approx \lambda\), this bound remains negligible even across more 
verifications than could feasibly be generated, without needing to 
refresh \(\ell\) or track rejections.

\section{The general approach
}\label{sec:general}

\subsection{GPV signatures}
Consider a general GPV-style signature.
Let \Module be a finitely generated module
over an integral domain \Ring:
in practice,
\Module is either 
a lattice (with \(\Ring = \ZZ\))
or 
a code (with \(\Ring = \FF_q\)).
Fix a cryptographic hash function
\(\hash: \{0,1\}^*\to \Module\).
A public key is a random-looking
\(\vec{M} = (\vec{M}_0,\ldots,\vec{M}_{n-1}) \in \Module^n\)
for some system parameter \(n\).
A signature on a message \(m\) under \(\vec{M}\)
is a tuple \(\sigma = (\salt,\vec{s})\)
with \(\salt \in \{0,1\}^{\lambda}\) (a random salt)
and \(\vec{s} \in \Ring^n\)
such that
\begin{equation}
    \label{eq:Gvalidation}
    \Constraint{\vec{s}}
    \qquad
    \text{and}
    \qquad
    \vec{s}\vec{M} := \sum\nolimits_{i=0}^{n-1}s_i\vec{M}_i
    =
    \hash(\salt\parallel m)
\end{equation}
where \(\Constraint{\vec{s}}\)
is a predicate on \(\vec{s}\)
such as having small norm (in~\Squirrels)
or a fixed number of nonzero entries (in~\Wave).
An important variant has
\hash mapping into (a subset of) \(\Ring^n\)
instead of \(\Module\),
and~\eqref{eq:Gvalidation}
is replaced by
\begin{equation}
    \label{eq:Gvalidation_alt}
    \Constraint{\vec{s}}
    \qquad
    \text{and}
    \qquad
    \big(\hash(\salt\parallel m) + \vec{s}\big)\vec{M}
    =
    \vec{0}
\,.
\end{equation}

\begin{example}
    In \Squirrels,
    \(\Ring = \ZZ\)
    and \(\Module = \ZZ/\Delta\)
    for some large \(\Delta\)
    (though later we lift to \(\Module = \ZZ\));
    verification uses~\eqref{eq:Gvalidation_alt}
    where
    \(\Constraint{\vec{s}}\)
    is \(\Vert\vec{s}\Vert_2^2 \le \lfloor\beta^2\rfloor\)
    for a small system parameter
    \(\beta\).
    For example,
    \Squirrels[-I] has \(n = 1034\),
    \(\Delta\) a 5048-bit modulus
    formed as the product of 165 31-bit primes,
    and \(\lfloor\beta^2\rfloor = 2026590\).
\end{example}

\begin{example}
    In \Wave,
    \(\Ring = \FF_3\)
    and \(\Module = \FF_3^{n-k}\)
    for system parameters \(n\) and \(k\);
    verification uses~\eqref{eq:Gvalidation}
    where \(\Constraint{\vec{s}}\)
    is \(\#\{i \mid s_i \not= 0\} = w\)
    for some (large) \(w < n\).
    For example:
    \Wave[822] 
has
    \(n = 8576\),
    \(k = 4288\),
    and \(w = 7668 \approx 0.9n\).
\end{example}

\subsection{Compressed verification}
\label{sec:sigma-comp}
Let \(\Sigma\) be a general GPV-style signature on an \Ring-module
\Module as described above,
and fix a set \KernelSet of \Ring-submodules of \Module.
We define a compressed-verification signature scheme
\Sigmacomp with the same \PKeyGen and \Sign as in \(\Sigma\),
but with \Verify replaced by the following
\CKeyGen, \VKeyGen, and \CVerify:
\begin{itemize}
    \item
        \CKeyGen:
        samples a random \Kernel from \KernelSet,
        and returns the quotient homomorphism
        \(\CK := \phi: \Module \to \RedModule := \Module/\Kernel\)
        (which is unique up to isomorphism, and has kernel \(\ker\phi =
        \Kernel\)).
    \item
        \VKeyGen:
        takes \(\PK = \vec{M}\)
        and returns \(\VK := (\phi,\overline{\vec{M}} :=
        (\phi(\vec{M}_1),\ldots,\phi(\vec{M}_n)))\).
    \item
        \CVerify:
        given \(m\), \(\sigma = (\salt,\vec{s})\),
        and \VK,
        let \(\vec{h} = \hash(\salt\parallel m)\).
        If \(\sigma\)
        would normally be verified with~\eqref{eq:Gvalidation},
        then \CVerify returns \Accept
        if and only if
        \(\Constraint{\vec{s}}\)
        and
        \(\phi(\vec{s}\vec{M}) = \phi(\vec{h})\);
        and since \(\phi\) is a homomorphism of \(\Ring\)-modules,
        this means
        \begin{equation}
            \label{eq:Lvalidation}
            \Constraint{\vec{s}}
            \qquad
            \text{and}
            \qquad
            \sum\nolimits_{i=0}^{n-1}s_i\overline{\vec{M}}_i
            =
            \phi(\vec{h})
            \quad
            \text{in }
            \RedModule
            \,.
        \end{equation}
        If verification normally uses~\eqref{eq:Gvalidation_alt},
        then \CVerify returns \Accept
        if and only if
        \begin{equation}
            \label{eq:Lvalidation_alt}
            \Constraint{\vec{s}}
            \qquad
            \text{and}
            \qquad
            \sum\nolimits_{i=0}^{n-1}(s_i + h_i)\overline{\vec{M}}_i
            =
            \vec{0}
            \quad
            \text{in }
            \RedModule
            \,.
        \end{equation}
\end{itemize}
If~\eqref{eq:Lvalidation_alt} is used for verification,
then \(\CK = \phi\) need not be included in \VK.

If the verifier verifies
many signatures from the same signer,
then the cost of \CKeyGen and \VKeyGen is amortised over the
many subsequent \CVerify calls.
A good choice of \(\phi\) can also reduce the time required for each verification.

\subsection{Correctness and security}
\label{sec:general-security}

The correctness of the signature scheme \Sigmacomp described above follows 
from \(\phi\) being a homomorphism:
\eqref{eq:Gvalidation} implies~\eqref{eq:Lvalidation}
and~\eqref{eq:Gvalidation_alt} implies~\eqref{eq:Lvalidation_alt}.
Our security goal is EUF-CMA (Existential 
unforgeability under adaptive chosen-message attacks).
Recall that a signature scheme $\Sigma = (\Gen, \Sign, \Vrfy)$ is 
said to be EUF-CMA if for all probabilistic polynomial-time (PPT) adversaries 
$\advA$, the probability that $\advA$ wins the following game is negligible 
in the security parameter $\lambda$:
\begin{enumerate}
    \item The challenger gets $(\mathsf{pk}, \mathsf{sk}) \leftarrow \Gen(1^\lambda)$ and gives $\mathsf{pk}$ to the adversary $\advA$.
    
    \item The adversary $\advA$ has access to a signing oracle $\mathcal{O}_{\Sign}$ and a verification oracle $\mathcal{O}_{\Verify}$. It may adaptively query  $\mathcal{O}_{\Sign}$ on messages $m_1, \dots, m_q $ and receive valid signatures $\sigma_i$ for each. It also may query $\mathcal{O}_{\Verify}$ on signatures $\sigma_i$ and know if it is valid signature or not.
    
    \item $\advA$ outputs a pair $(m^*, \sigma^*)$.
    
    \item $\advA$ wins the game if:
    \begin{enumerate}
        \item $\Vrfy_{\mathsf{pk}}(m^*, \sigma^*) = 1$ (i.e., the signature is valid), and
        \item $m^* \notin \{m_1, \dots, m_q\}$ (i.e., $m^*$ was not queried to the signing oracle).
    \end{enumerate}
\end{enumerate}

We want to relate the EUF-CMA security
of a GPV signature with compressed verification
to the assumed EUF-CMA security of the original scheme.
We model forgery attempts as attempts at solving a hidden-structure problem:

\begin{definition}[Submodule Element Guessing Problem \(\SGP(\KernelSet,\TSet)\)]
    \label{def:submod-guessing}
    Let \Module be an \Ring-module,
    \TSet a finite subset of \Module,
    and \KernelSet a set of \Ring-submodules of~\Module.
    The \textbf{Submodule Element Guessing Problem}
    \(\SGP(\KernelSet,\TSet)\) is:
    given a membership oracle \( \mathcal{O}_\Kernel \) 
    for an unknown submodule \(\Kernel \in \KernelSet\)
    taking input in \TSet
    (i.e.: \(\mathcal{O}_\Kernel\) takes \( \vec{t} \in \TSet \) and returns
    \True if \(\vec{t} \in \Kernel\) and \False otherwise),
find an element \( \vec{t}^* \not= \vec{0} \in \TSet \)
    such that \( \mathcal{O}_\Kernel(t^*) = \True \),
    i.e., \( \vec{t}^* \in \Kernel \).
\end{definition}

The set \KernelSet in Definition~\ref{def:submod-guessing}
represents the set of possible kernels of the secret homomorphism \(\phi\),
and \TSet represents the vectors constructed by forgery attempts
before input to \(\phi\).

\begin{definition}
    \label{def:TSet}
    Let \(\Sigma\) be a general GPV signature scheme
    as above.  We define
    \begin{align*}
        \TSet(\Sigma)
        & := 
        \big\{
            \vec{s}\vec{M} - \vec{h}
            :
            \vec{h} \in \operatorname{Im}(\hash),
            \vec{s} \in \Ring^n
            \mid 
            \Constraint{\vec{s}}
        \big\}
        \intertext{if \(\Sigma\) uses~\eqref{eq:Gvalidation} for verification, or}
        \TSet(\Sigma)
        & := 
        \big\{
            (\vec{h} + \vec{s})\vec{M}
            : 
            \vec{h} \in \operatorname{Im}(\hash),
            \vec{s} \in \Ring^n
            \mid 
            \Constraint{\vec{s}}
        \big\}
    \end{align*}
    if \(\Sigma\) uses~\eqref{eq:Gvalidation_alt} for verification.
\end{definition}

\begin{theorem}
\label{thm:unforgeability}
    Let \(\Sigma\) be a general GPV-style signature scheme
    on an \Ring-module \Module,
    and fix a set \KernelSet of \Ring-submodules of \Module.
    For each \Kernel in \KernelSet,
    let \Sigmacomp[\Kernel] be the instance of \Sigmacomp
    where \CKeyGen samples \Kernel from \KernelSet.
    If \(\mathcal{A}\) is an algorithm running in time \(T\)
    that wins 
    the EUF-CMA game for \Sigmacomp[\Kernel]
    with probability \(P\),
    then there exists an algorithm \(\mathcal{B}\)
    running in time \(T+O(1)\)
    that succeeds with probability \(P\)
    in winning the EUF-CMA game for \(\Sigma\),
    or solving
the \(\SGP(\KernelSet,\TSet(\Sigma))\) instance corresponding to \Kernel.
\end{theorem}
\begin{proof}
    The challenger for EUF-CMA game of $\Sigma$ gives \(\PK = \vec{M}\) to the algorithm \(\mathcal{B}\). \(\mathcal{B}\) has access to a signing oracle $\mathcal{O}_{\Sign}$, a verification oracle $\mathcal{O}_{\Vrfy}$ both associated to $\Sigma$ and a submodule membership oracle \( \mathcal{O}_K \) associated to \(\SGP(\KernelSet,\TSet(\Sigma))\). \(\mathcal{B}\) calls \(\mathcal{A}\) on \(\PK = \vec{M}\). \(\mathcal{B}\) answers a signing oracle query $m$  from \(\mathcal{A}\) by $\mathcal{O}_{\Sign}(m)$ as the signing process is the same for \(\Sigma\) and  \Sigmacomp[\Kernel]. If \(\mathcal{A}\) makes a query to the verification oracle for \Sigmacomp[\Kernel] on a signature $\sigma=(\salt,\vec{s})$, $\mathcal{B}$ queries $\mathcal{O}_{\Vrfy}(\sigma)$. If $\mathcal{O}_{\Vrfy}(\sigma)=1$ then $\mathcal{B}$ answers $1$ to $\mathcal{A}$. Otherwise, If \(\Sigma\) verifies with~\eqref{eq:Gvalidation} then \(\mathcal{B}\) sets  \(\vec{t} := \vec{s}\vec{M} - \hash(\salt\parallel m)\); if \(\Sigma\) verifies with~\eqref{eq:Gvalidation_alt} instead, then \(\mathcal{B}\) sets \(\vec{t} := (\hash(\salt\parallel m) + \vec{s})\vec{M}\). \(\mathcal{B}\) answers the verification oracle query by \( \mathcal{O}_K(\vec{t}) \).

    If \(\mathcal{A}\) fails,
    then \(\mathcal{B}\) fails.
    Otherwise, it receives a \((m^*,\sigma^* = (\salt^*,\vec{s}^*))\)
    such that \(\Constraint{\vec{s}^*}\) holds
    and \(\CVerify(m^*,\sigma^*,\VK) = \Accept\).
    $\mathcal{B}$ computes $\vec{t}^*$,
    if \(\vec{t}^* = \vec{0}\),
    then \(\Verify(m^*,\sigma^*,\PK)\) would return \Accept,
    so \(\mathcal{B}\) returns \((m^*,\sigma^*)\) and wins the EUF-CMA game associated to $\Sigma$.
    Otherwise,
    \(\vec{t}^*\) is a nonzero element of \(\Kernel\),
    so \(\mathcal{B}\) returns \(\vec{t}^*\) and finds a solution to \(\SGP(\KernelSet,\TSet(\Sigma))\).
    \qed
\end{proof}

Theorem~\ref{thm:unforgeability}
tells us that if \(\Sigma\) is EUF-CMA secure,
then forging a signature for \Sigmacomp[\Kernel]
is \emph{at least as hard} as solving the
\(\SGP(\KernelSet,\TSet(\Sigma))\) instance
corresponding to \Kernel.
We can therefore choose secure parameters for \Sigmacomp
by choosing the set \KernelSet of compression keys
such that random instances of \(\SGP(\KernelSet,\TSet(\Sigma))\) are hard.

\subsection{The hardness of \SGP}
\label{sec:general-hardness}

The hardness of \(\SGP(\KernelSet,\TSet)\)
depends on the choice of \KernelSet and \TSet,
and also on the properties of \Ring and \Module 
(and \(\Module/\Kernel\) for \Kernel in \KernelSet).
There are a few general things that we can say
before returning to the problem in the concrete cases of \Squirrels and
\Wave later.
Let \(\mathcal{O}\) be a membership oracle for a secret \Kernel
sampled uniformly random from \KernelSet,
and accepting only queries from \TSet;
and let \(\phi : \Module \to \Module/\Kernel\)
be the quotient homomorphism.
We can assume \(\phi(\TSet) = \Module/\Kernel\).
The adversary's goal is to find some \(\vec{t} \not= \vec{0}\) in \TSet
such that \(\mathcal{O}(\vec{t}) = \True\):
implicitly,
\(\vec{t} \in (\Kernel\setminus\{\vec{0}\})\cap\TSet\).

The adversary makes a series of adaptive queries 
\(\vec{t}^{(1)},\vec{t}^{(2)},\ldots\) to \(\mathcal{O}\).
Suppose \(\mathcal{O}(\vec{t}^{(i)})=\False\)
for \(1 \le i \le Q\).
The adversary wants \(\phi(\vec{t}^{(Q+1)}) \not= \phi(\vec{t}^{(i)})\)
for \(1 \le i \le Q\),
because none of the \(\phi(\vec{t}^{(i)}\) were \(\vec{0}\).
In the best case for the adversary
all these values are distinct,\footnote{Note that collisions \(\phi(\vec{t}^{(i)}) = \phi(\vec{t}^{(j)})\) 
    are not useful to the adversary, because they cannot detect them without
    querying \(\mathcal{O}\) on all the \(\vec{t}^{(i)} - \vec{t}^{(j)}\).
}
so if the adversary chooses \(\vec{t}^{(Q+1)}\) arbitrarily then
\[
    P\big[ \mathcal{O}(\vec{t}^{(Q+1)}) \big]
    \le
    \frac{
        1
    }{
        \#(\Module/\Kernel)
        -
        Q
    }
    \,.
\]
(If the adversary can choose \(\vec{t}^{(Q+1)}\)
such that it maps into a proper submodule \(\mathcal{N} \subset
\Module/\Kernel\) then we can replace \(\#(\Module/\Kernel)\)
with \(\#\mathcal{N}\) and \(Q\) with the number of prior queries
landing in \(\mathcal{N}\), but this is not significant in our
applications.)

In the meantime, the adversary has also learned
that
\[
    \Kernel \notin \bigcup_{i=1}^{Q}\KernelSet_{\vec{t}^{(i)}}
    \quad
    \text{where}
    \quad
    \KernelSet_\vec{t}
    :=
    \{\Kernel\in\KernelSet\mid\vec{t}\in\Kernel\}
    \subset
    \KernelSet
    \,.
\]
In our applications \KernelSet is finite,
so there exists an integer
\[
    \kappa_\TSet
    :=
    \max\{\#\KernelSet_\vec{t} : \vec{t}\in\TSet\}
    \,;
\]
each unsuccessful query eliminates up to \(\kappa_\TSet\)
candidate kernels from consideration.
The adversary must choose
\(\vec{t}^{(Q+1)}\)
such that
\(
    \KernelSet_{\vec{t}^{(Q+1)}} 
    \not\subset 
    \bigcup_{i=1}^{Q}\KernelSet_{\vec{t}^{(i)}}
\)
to have any chance of success.
If \(\vec{t}^{(Q+1)}\) is chosen arbitrarily
among the elements of \TSet
such that 
\(
    \KernelSet_{\vec{t}^{(Q+1)}} 
    \setminus
        \bigcup_{i=1}^{Q}\KernelSet_{\vec{t}^{(i)}}
\)
is maximal,
then the probability of success
is
\begin{equation}
    \label{ineq:Sx}
    P\big[ \mathcal{O}(\vec{t}^{(Q+1)}) \big]
    \le
    \frac{
        \#\big(
            \KernelSet_{\vec{t}^{(Q+1)}}
            \setminus
            \bigcup_{i=1}^{Q}\KernelSet_{\vec{t}^{(i)}}
        \big)
    }{
        \#\KernelSet 
        -
        \# \bigcup_{i=1}^{Q}\KernelSet_{\vec{t}^{(i)}}
    }
    \le
    \frac{
        \kappa_\TSet
    }{
        \#\KernelSet - \kappa_\TSet Q
    }
    \,.
\end{equation}

For EUF-CMA security with parameter \(\mu\) against this adversary,
we need to ensure the success probability after \(Q\) queries
is \(\le 2^{-\mu}\),
so
\begin{equation}
    \label{eq:query-limit}
    \min(\#\KernelSet/\kappa_\TSet,\#(\Module/\Kernel))
    \ge
    2^\mu + Q
    \,.
\end{equation}
We should therefore sample \Kernel from an \KernelSet
chosen
such that (at least)
\begin{equation}
    \label{eq:KernelSet-bounds}
    \#\KernelSet/\kappa_\TSet \ge 2^\mu
    \qquad
    \text{and}
    \qquad
    \#(\Module/\Kernel) \ge 2^\mu
    \text{ for each }
    \Kernel \in \KernelSet
    \,,
\end{equation}
and we should replace the verification key
when the number of (failed) compressed verifications approaches \(2^\mu\).

Heuristically,
we will suppose that the adversary cannot
improve on any strategy that simply enumerates queries \(\vec{t}^{(i+1)}\)
while maximising 
\(\#\KernelSet_{\vec{t}^{(i+1)}}\setminus \bigcup_{j=1}^{i}\KernelSet_{\vec{t}^{(j)}}\).
Indeed, to do so they would need more information
on \(\phi(\vec{x})\) for unqueried \(\vec{x}\).
Our heuristic is that this information has to come from exploiting the
\Ring-module structures of \Module and \(\Module/\Kernel\),
but these structures generally do not help.

First,
\(
    (\phi(\vec{x}) \not= \vec{0})\wedge(\phi(\vec{y}) \not= \vec{0})
    \centernot\implies
    \phi(\vec{x}+\vec{y}) \not= \vec{0}
\)
unless \(\vec{y} \in \Ring\vec{x}\) or \(\vec{x}\in \Ring\vec{y}\).
This tells us that given the results of
\(\mathcal{O}(\vec{t}^{(i)})\) for \(1 \le i \le Q\),
we cannot predict the result of \(\mathcal{O}(\vec{t}^{(Q+1)})\)
for general linear combinations
\(\vec{t}^{(Q+1)} = \sum_{i=1}^Q\alpha_i\vec{t}^{(i)}\).

Looking at scalar multiplication,
there are two cases.
\begin{enumerate}
    \item
        If \(\alpha \not= 0 \in \Ring\) is invertible on \(\Module/\Kernel\),
        then
        \(\phi(\vec{x}) = \vec{0} \iff \phi(\alpha\vec{x}) = \vec{0}\)
        for all \(\vec{x}\).
        In this case, if we know \(\mathcal{O}(\vec{x}) = \False\),
        then we can predict that \(\mathcal{O}(\alpha\vec{x}) = \False\)
        \emph{without} querying \(\mathcal{O}\) on \(\alpha\vec{x}\)
        (and vice versa).
        But these ``virtual'' queries cannot not help the adversary,
        because
        \(\KernelSet_{\alpha\vec{x}} = \KernelSet_{\vec{x}}\)
        for all such \(\alpha\).
    \item
        If \(\alpha \not= 0 \in \Ring\) is \emph{not} invertible on
        \(\Module/\Kernel\),
        then 
        \(\phi(\vec{x}) = \vec{0} \implies \phi(\alpha\vec{x}) = \vec{0}\)
        for all \(\vec{x}\),
        but the converse does not hold;
        likewise, \(\KernelSet_{\vec{x}} \subset \KernelSet_{\alpha\vec{x}}\)
        but the inclusion may be strict.
        If such elements \(\alpha\) are known,
        then the adversary should query on \(\alpha\vec{x}\)
        instead of \(\vec{x}\) to maximise
        \(\#\KernelSet_{\alpha\vec{x}}\)
        (and hence \(\#\KernelSet_{\vec{t}^{(i+1)}}\setminus
        \bigcup_{j=1}^{i}\KernelSet_{\vec{t}^{(j)}}\)).
        If these \(\alpha\) exist but are \emph{not} known to the
        adversary,
        then they should try to guess them in order to approach the
        ideal bound of~\eqref{ineq:Sx}.
\end{enumerate}
In our application to \Wave,
\(\Ring = \FF_3\) is a field,
so we are always in the first situation.
For Rabin--Williams and \Squirrels,
\(\Ring = \ZZ\)
and \(\Module/\Kernel = \ZZ/d\ZZ\)
for some \(d\) unknown to the adversary.
Every query is \(\alpha\cdot 1\)
for some \(\alpha\),
and the adversary's goal is precisely to find \(\alpha \not= 0\)
divisible by the unknown \(d\):
that is, they are (or are trying to be) in the second situation.

\begin{remark}
    The bounds in~\eqref{eq:KernelSet-bounds}
    may be overly pessimistic:
    even if \(\SGP(\KernelSet,\TSet(\Sigma))\) is hard,
    forging signatures in \Sigmacomp[\Kernel]
    for random \Kernel in \KernelSet may be significantly harder.
With GPV-style signatures,
    a solution \(\vec{t}\)
    to the \(\SGP(\KernelSet,\TSet(\Sigma))\) instance for \Kernel
    gives a forgery against \Sigmacomp[\Kernel]
    only if we can construct \((m,\sigma=(\salt,\vec{s}))\)
    mapping to \(\vec{t}\);
    and this is made difficult by the
    need to satisfy \Constraint{\vec{s}}.
    After all, if we could find \((m,\sigma)\)
    for arbitrary \(\vec{t}\)
    then we could find them for \(\vec{0}\),
    and thus construct forgeries
    in the original scheme \(\Sigma\).
    For a computationally bounded adversary,
    it may be infeasible to construct \((m,\sigma)\)
    yielding implicit queries \(\vec{t}\) with large
    \(\#\KernelSet_\vec{t}\),
    which means the success probability is actually much lower---and then
    we can make \KernelSet (and thus, potentially, \(|\VK|\))
    much smaller.
    Hence,
    while setting parameters to make \SGP hard 
    will guarantee unforgeability,
    these parameters may also be much larger than
    what is required for EUF-CMA in practice.
\end{remark}

\section{Compressed verification for \Squirrels }\label{sec:sq_mar}

Now we turn our attention to \Squirrels.
The challenge here is to define compressed verification algorithms
that, like \Squirrels, avoid multiprecision arithmetic.
To simplify presentation,
we use the following notation:

\begin{definition}
    Given a list of primes \(\vec{m} = (m_1,\ldots,m_s)\),
    we write
    \[
        \RNS{x}{\vec{m}}
        :=
        \big(x \bmod{m_1},\ldots,x \bmod{m_s}\big)
        \quad
        \text{for all }
        x \in \ZZ
        \,.
    \]
\end{definition}

\subsection{The \Squirrels signature scheme}
\label{sec:Squirrels}
\squirrels is a GPV signature on
unstructured lattices. More precisely,
it uses co-cyclic lattices:
\(n\)-dimensional lattices \Lattice
such that \(\ZZ^n/\Lattice=\ZZ/\Delta\ZZ\)
for some \(\Delta\).
Co-cyclic lattices are dominant among
full-rank integer lattices~\cite{Nguyen--Shparlinski}:
their natural density is \(\approx 85\%\).
\Squirrels
works with co-cyclic lattices \Lattice
of determinant
\[
    \Delta := \squirrelsprime_1\cdots \squirrelsprime_\squirrelsprimecount
\]
where
\(
    \squirrelsprimes
    =
    (\squirrelsprime_1,\ldots,\squirrelsprime_\squirrelsprimecount)
\)
is a fixed tuple
of 31-bit primes
(the length \squirrelsprimecount
depends on the security parameter).
Table~\ref{tab:sizes_squirrel} gives
the \Squirrels parameter sets.

\begin{table}[htp]
    \centering
    \caption{Parameters for \squirrels instances.}
    \label{tab:sizes_squirrel}
    \resizebox{\textwidth}{!}{
    \begin{tabular}{l|rrrrr}
        \toprule
        NIST Security Level
        & 1      & 2     & 3     & 4      & 5
        \\
        \midrule
        Lattice dimension $n$ &
        $1034$   & $1164$   & $1556$    & $1718$    & $2056$
        \\
        Hash space size \(q\) & 4096 & 4096 & 4096 & 4096 & 4096
        \\
        Max. signature square norm \(\lfloor\beta^2\rfloor\)
        & 2\,026\,590 & 2\,442\,439 & 4\,512\,242 & 3\,659\,372 &
        5\,370\,115
        \\
        Number \squirrelsprimecount of small primes
        & 165 & 188 & 262 & 275 & 339
        \\
        Bitlength of \(\Delta\)
        & 5048 & 5738 & 8017 & 8402 & 10347
        \\
        \midrule
        Signature Size ($\SI{}{\byte}$)             & $1019$   & $1147$   & $1554$    & $1676$    & $2025$
        \\
        Public Key Size ($\SI{}{\byte}$)     & $681\,780$ & $874\,576$ & $1\,629\,640$ & $1\,888\,700$ & $2\,786\,580$
        \\
        \bottomrule
    \end{tabular}
    }
\end{table}

\squirrels is built on the one-way
function 
\begin{equation}
    \begin{split}
        f : D_n & \longrightarrow \ZZ/\Delta\ZZ \\
        \vec{x} & \longmapsto \vec{x}\vec{A}^T \pmod \Delta,
    \end{split}
\end{equation}
where $\vec{A}$ is the matrix defining \Lattice
and \(D_n = \{\vec{e} \in \Z^n \mid \Vert\vec{e}\Vert \leq \beta\}\).
One-wayness depends on the hardness of
$GSIS_{n,\Delta,\beta}$,
that is,
finding a vector \(\vec{x}\)
such that \(\vec{x}\vec{A}^T \equiv 0 \pmod \Delta\)
and \(\Vert\vec{x}\Vert \leq \beta\) for some small \(\beta\).
\change{should we introduce the problem GSIS? B: that sentence defines GSIS.}

Suppose \Lattice is co-cyclic of dimension \(n\)
and determinant \(\Delta\).
We can specify \Lattice
with the row-HNF of its generating matrix, which
is determined by a vector
\[
    \vcheck = (\vchecki[1],\ldots,\vchecki[n])
    \in
    (\ZZ/\Delta\ZZ)^n
    \quad
    \text{with}
    \quad
    \vchecki[n] = -1
    \,.
\]
The public key encodes \vcheck
as a list of lists of residues
mod the small primes:
\[
    \pk
    =
    \big(
        ( v_{i,j} := \vchecki \bmod \squirrelsprime_j)_{i=1}^{n-1}
    \big)_{j=1}^\squirrelsprimecount
\]
(since \(\vchecki[n] = -1\) by convention, 
there is no need to store it or any of the~\(v_{n,j}\)).
The \(v_{i,j}\) are encoded as \emph{signed} twos-complement 32-bit integers,
but they are all non-negative 
(except the \(v_{n,j}\), which are all \(-1\) and not stored
anyway);
in particular, \(0 \le v_{i,j} < 2^{31}\) for all \(1 \le i < n\)
and \(1 \le j \le s\).

The private key encodes a ``good'' basis
for \Lattice,
which allows sampling short vectors in \Lattice
following a Gaussian distribution using Klein's trapdoor sampler.
\PKeyGen ensures that 
each \(\vchecki\)
looks like a uniform random integer modulo~\(\Delta\).

We now focus on \Squirrels verification
(\PKeyGen and \Sign are detailed
in~\cite{squirrels_sub}).
The verifier accepts $\sigma = (\vec{s},\salt)$
if two conditions are met:
\begin{enumerate}
    \item  The vector $\vec{s}$
        is short:
        $\Vert\vec{s}\Vert \leq \beta$,
        which is more easily checked as
        \[
            \label{eq:squirrels-norm-condition}
            \Constraint{\vec{s}}:
            \Vert\vec{s}\Vert^2 \le \lfloor\beta^2\rfloor
            \,.
        \]
    \item The vector $\vec{c} := \vec{s} + \vec{h}$
        (where \(\vec{h} := \hash(\salt\parallel\vec{m})\))
        is in \Lattice.
        That is,
        \begin{equation}
            \label{eq:squirrels-verification}
            \sum\nolimits_{i=1}^{n} c_i\vchecki
            \equiv
            0
            \pmod{\Delta}
            \,,
        \end{equation}
        or equivalently (by the CRT)
        \begin{equation}
            \label{eq:squirrels-verification-CRT}
            \sum\nolimits_{i=1}^n c_iv_{i,j}
            \equiv 0
            \pmod{\squirrelsprime_j}
            \qquad
            \text{for all }
            1 \le j \le \squirrelsprimecount
            \,.
        \end{equation}
\end{enumerate}
We can thus verify by
checking~\eqref{eq:squirrels-verification-CRT}
for each of the \(\squirrelsprime_j\) in turn,
as in Algorithm~\ref{alg:verify}.

\begin{algorithm}[H]
 \scriptsize
    \Parameters{$q$, $n$, $\lfloor{\beta^2}\rfloor$,
        and $P_\Delta = (\squirrelsprime_1,\ldots,\squirrelsprime_m)$
    }
    \KwIn{Signature $\sigma= (\salt, \underline{\vec{s}})$,
        message $m$,
        public key $\pk = ((v_{i,j})_{i=1}^{n-1})_{j=1}^m$ with
        \(v_{i,j} = \vchecki \bmod{\squirrelsprime_j}\)
    }
    \KwOut{\Accept if \(\sigma\) is a valid signature on \(m\) under
    \pk, otherwise \Reject.}
$\vec{s} \gets \decomp(\underline{\vec{s}})$
    \;
\If{$\vec{s} = \perp$ \textbf{or} 
        $\Vert\vec{s}\Vert_2^2 > \lfloor \beta^2 \rfloor$}{
        \Return \Reject
    }
$\vec{c} \gets \vec{s} + \htopoint(m\parallel\salt, q, n)$
    \;
\ForEach(\tcp*[f]{Trivially parallelizable}){$1 \le j \le m$}{
        $S \gets \sum_{i=0}^{n-1} c_i v_{i,j} \bmod \squirrelsprime_j$
        \;
        \If(\tcp*[f]{Uses \(\vchecki[n] = -1\)}){
            $S - c_n \neq 0 \mod \squirrelsprime_j$
        }{
            \Return \Reject
        }
    }
    \Return \Accept
    \;
\caption{Verification algorithm for \squirrels. \label{alg:verify}}
\end{algorithm}

\subsection{Homomorphisms for \Squirrels verification}

\Squirrels is an instance of our general framework
with \(\Ring = \ZZ\)
and \(\Module = \ZZ/\Delta\ZZ\).
But the only homomorphisms from \(\ZZ/\Delta\ZZ\)
map through \(\ZZ/\Delta'\ZZ\) for \(\Delta'|\Delta\),
and while a verifier could choose a secret \(\phi\)
by choosing a secret subset of the \(\squirrelsprime_j\),
an adversary who could forge a signature for a large subset
of the \(\squirrelsprime_j\) would fool many verifiers.
There are many more homomorphisms from \(\ZZ\),
and we can lift \Squirrels trivially to \(\Module = \ZZ\)
if we replace \eqref{eq:squirrels-verification}
with the equivalent condition
\begin{equation}
    \label{eq:squirrels-verification-lift}
    \sum\nolimits_{i=1}^{n} c_i\cdot\vchecki
    =
    k\Delta
    \qquad
    \text{for some integer } k
    \,.
\end{equation}

Let the verifier choose a secret list of secret 31-bit primes 
\(\secretprimes = (\secretprime_1,\ldots,\secretprime_\secretprimecount)\),
each prime to \(\Delta\).
The parameter \(t\) is a function of the desired verification security
level,
to be determined later in~\S\ref{sec:chip-security}.
Now the verifier could check
\begin{equation}
    \label{eq:squirrels-verification-reduced}
    \sum\nolimits_{i=1}^{n} c_i\cdot(\vchecki\bmod\secretprime_j)
    \equiv
    k(\Delta\bmod{\secretprime_j})
    \pmod{\secretprime_j}
    \quad
    \text{for }
    1 \le j \le \secretprimecount
\end{equation}
---but \(k\) is not included in the signature,
and re-computing it is the same computation as a full \Squirrels
verification.
Instead, we will verify by implicitly recovering \(k\)
modulo \(\prod_j\secretprime_j\),
and checking that it is in the appropriate range.

First, we need to compute
each of the \(\RNS{\vchecki}{\secretprimes}\)
from the
\(
    \RNS{\vchecki}{\squirrelsprimes}
    =
    (v_{i,j} 
    = 
    \vchecki
    \bmod{\squirrelsprime_j})_{j=1}^\squirrelsprimecount
\).
In the spirit of \Squirrels,
we want to avoid multiprecision integer arithmetic,
so we need to compute the \(\RNS{\vchecki}{\secretprimes}\)
\emph{without reconstructing} any of the integers \vchecki.
Our main tool is the \emph{explicit} CRT.\footnote{This is similar to the modular reduction in RNS arithmetic
    in~\cite{BernsteinS07},
    but 
    there,
    \squirrelsprimes can be freely chosen to optimise
    computations on the operands;
    here, \squirrelsprimes is fixed.
}

\begin{definition}
    With the notation above:
    for each \(1 \le i \le \squirrelsprimecount\),
    we write \( \Delta_i\) for \(\Delta/\squirrelsprime_i \),
    and let \( q_i\) be the the unique integer in
            \((1,\squirrelsprime_i)\)
            such that 
            \(q_i\Delta_i \equiv 1 \pmod{\squirrelsprime_i}\).
\end{definition}
We will never explicitly compute with the \(\Delta_i\)
(they are a notational convenience).
We \emph{will} need the \(q_i\),
and these can be precomputed in advance
(using e.g.~Algorithm~\ref{alg:qCoefficients}
in Appendix~\ref{app:rns}).
Each is a positive 31-bit integer.

\begin{lemma}[Explicit CRT]
    \label{lemma:explicit-CRT}
    If \(0 \le x < \Delta\)
    and
    \(
        (x_1,\ldots,x_\squirrelsprimecount)
        =
        \RNS{x}{\squirrelsprimes}
    \),
    then
    \begin{equation}
        \label{eq:explicit-CRT}
        x
        =
        \alpha \Delta - \lfloor\alpha\rfloor \Delta
        \quad
        \text{where}
        \quad
        \alpha
        =
        \sum\nolimits_{i=1}^\squirrelsprimecount x_iq_i/\squirrelsprime_i
        \,.
    \end{equation}
\end{lemma}
\begin{proof}
    Observe that \(\alpha\) is a rational number,
    but \(\alpha \Delta\) is an integer.
    The CRT says \(x \equiv \alpha \Delta \pmod{\Delta}\),
    so obviously
    \(\alpha \Delta - \lfloor\alpha\rfloor \Delta \equiv x \pmod{\Delta}\).
    But
    \(0 \le \alpha - \lfloor\alpha\rfloor < 1\)
    by construction,
    so \(0 \le \alpha \Delta - \lfloor\alpha\rfloor \Delta < \Delta\),
    so \(\alpha \Delta - \lfloor\alpha\rfloor \Delta = x\).
    \qed
\end{proof}

Lemma~\ref{lemma:explicit-CRT}
gives an exact expression for \(0 \le x < \Delta\)
in terms of \(\RNS{x}{\squirrelsprimes}\)
that we can use to compute 
\(\RNS{x}{\secretprimes}\)
by computing the integers \(\alpha \Delta\) and \(\lfloor\alpha\rfloor \Delta\)
modulo each \(\secretprime_j\).
We precompute
the \(q_i\),
\(\RNS{\Delta_i}{\secretprimes}\),
and \(\RNS{\Delta}{\secretprimes}\).
Computing \(\alpha \Delta = \sum_ix_iq_i\Delta_i\)
modulo \(\secretprime_j\) is straightforward.
The interesting part is determining
\(\lfloor\alpha\rfloor\),
and thus computing \(\lfloor\alpha\rfloor \bmod{\secretprime_j}\),
without computing~\(\alpha\).
We will do this
using fixed-point approximations,
as in~\cite{BernsteinS07}.
Lemma~\ref{lemma:explicit-floor},
an adaptation of~\cite[Lemma 3.1]{BernsteinS07},
shows that with a relatively low precision
we can 
determine \(\lfloor\alpha\rfloor\)
up to a possible error of \(1\).

\begin{lemma}
    \label{lemma:explicit-floor}
    Let \(\alpha_1, \ldots, \alpha_\squirrelsprimecount\)
    be non-negative real numbers,
    and set \(\alpha := \sum_{j=1}^\squirrelsprimecount \alpha_j\).
    Fix some integer \(a \ge \log_2\squirrelsprimecount + 1\).
    Then
    \[
        f :=
        \Big\lfloor
            \frac{s}{2^a}
            +
            \frac{1}{2^a}
            \sum\nolimits_{j=1}^\squirrelsprimecount\big\lfloor
                2^a \alpha_j
            \big\rfloor
        \Big\rfloor
        \quad
        \text{is either }
        \lfloor\alpha\rfloor
        \text{ or }
        \lfloor\alpha\rfloor + 1\,.
    \]
    Further,
    if \(\alpha - \lfloor\alpha\rfloor < 1 - \squirrelsprimecount/2^a\)
    then \(f\) is exactly \(\lfloor\alpha\rfloor\).
\end{lemma}
\begin{proof}
    Note that \(2^a \ge 2\squirrelsprimecount\),
    so \(\squirrelsprimecount/2^a \le 1/2\).
    Let \(q := (1/2^a)\sum_j\lfloor 2^a\alpha_j\rfloor\).
    By construction,
    \(0 \le 2^a\alpha_j - \lfloor 2^a\alpha_j \rfloor < 1\)
    for each \(1 \le j \le \squirrelsprimecount\).
    Summing over \(j\) gives
    \(0 \le 2^a\alpha - 2^a q < \squirrelsprimecount\),
    so
    \(0 \le \alpha - q < \squirrelsprimecount/2^a\);
    so \(\lfloor\alpha\rfloor < \squirrelsprimecount/2^a + q \le
    \squirrelsprimecount/2^a + \alpha\).
    Taking floors gives
    \(\lfloor\alpha\rfloor \le f \le \lfloor \squirrelsprimecount/2^a + \alpha\rfloor\).
    But
    \(\lfloor \squirrelsprimecount/2^a + \alpha \rfloor\)
    is either \(\lfloor \alpha\rfloor\)
    or \(\lfloor\alpha\rfloor + 1\),
    because \(0 < \squirrelsprimecount/2^a \le 1/2\),
    proving the first statement.
    For the second,
    if \(\alpha - \lfloor\alpha\rfloor < 1 - \squirrelsprimecount/2^a\)
    then \(\alpha + \squirrelsprimecount/2^a < \lfloor\alpha\rfloor+1\),
    so \(\lfloor\alpha + \squirrelsprimecount/2^a\rfloor = \lfloor\alpha\rfloor\),
    and thus \(f = \lfloor\alpha\rfloor\).
    \qed
\end{proof}

\begin{theorem}
    \label{lemma:explicit-CRT-mod}
    Fix
    an integer \(a \ge \log_2(\squirrelsprimecount) + 1\).
    Given
    \(\secretprimes\),
    \(\RNS{\Delta}{\secretprimes}\),
    \((\RNS{\Delta_i}{\secretprimes})_{i=1}^s\),
    and
    \(\RNS{x}{\vec{p}}\)
    for some \(0 \le x < \Delta\),
    Algorithm~\ref{alg:ModECRT}
    returns \(\RNS{z}{\secretprimes}\)
    where \(z\) is either \(x\) or \(x - \Delta\).
    Further:
    if \(x < (1 - \squirrelsprimecount/2^a)\Delta\),
    then \(z = x\).
\end{theorem}
\begin{proof}
    Algorithm~\ref{alg:ModECRT}
    evaluates~\eqref{eq:explicit-CRT}
    modulo \(\secretprime_j\)
    for \(1 \le j \le \secretprimecount\),
    computing the floor using
    Lemma~\ref{lemma:explicit-floor}
    with \(\alpha_i = x_iq_i/\squirrelsprime_i\) 
    for \(1\le i \le \squirrelsprimecount\).
    \qed
\end{proof}

\begin{algorithm}[H]
 \scriptsize
    \caption{Explicit CRT: computing \(\RNS{x}{\secretprimes}\)
        or \(\RNS{x-\Delta}{\secretprimes}\) 
        from \(\RNS{x}{\squirrelsprimes}\)
    }
    \label{alg:ModECRT}
\Parameters{\Squirrels prime vector \(\vec{\squirrelsprime} =
        (\squirrelsprime_1,\ldots,\squirrelsprime_\squirrelsprimecount)\)
        and
        CRT coefficients \((q_1,\ldots,q_\squirrelsprimecount)\);
        an integer \(a \ge \log_2\squirrelsprimecount + 1\)
        (fixed-point precision)
    }
\KwIn{Secret prime list \(\secretprimes\);
        \(\RNS{\Delta}{\secretprimes}\),
        \((\RNS{\Delta_i}{\secretprimes})_{i=1}^s\),
        and
        \(\RNS{x}{\vec{\squirrelsprime}}\) for \(0 \le x < \Delta\)
    }
    \KwOut{\(\RNS{z}{\secretprimes}\),
        where \(z = x\) or \(x - \Delta\)
        (if \(x < (1 - s/2^a)\Delta\), then \(z = x\))
    }
\Function{\ModECRT{\(\secretprimes\),
    \(\RNS{\Delta}{\secretprimes}\),
    \((\RNS{\Delta_i}{\secretprimes})_{i=1}^\squirrelsprimecount\),
    \(\RNS{x}{\squirrelsprimes}\)}}{
\((z_1,\ldots,z_t) \gets (0,\ldots,0)\)
        \;
\(f \gets \squirrelsprimecount\)
        \;
\ForEach{\(1 \le j \le \squirrelsprimecount\)}{
            \(y_j \gets x_j\cdot q_j\)
            \;
\(f \gets f + \lfloor 2^a y_j/\squirrelsprime_j\rfloor\)
            \label{alg:ModECRT:numerator}
            \tcp*{See Remark~\ref{rem:numerator-floor-main}}
\ForEach{\(1 \le k \le \secretprimecount\)}{
                \(y_{j,k} \gets y_j \bmod \secretprime_k\)
                \;
\label{alg:ModECRT:z-update}
                \(
                    z_k
                    \gets
                    (z_k + y_{j,k}\cdot \Delta_{j,k}) \bmod{\secretprime_k}
                \)
            }
}
        \label{alg:ModECRT:outer-floor}
        \(f \gets \lfloor f/2^a \rfloor\)
        \tcp*{Right-shift \(f\) by \(a\)}
\ForEach{\(1 \le k \le t\)}{
            \(z_k \gets (z_k - f\cdot\Delta_k) \bmod{\secretprime_k}\)
        }
\Return{\((z_1,\ldots,z_t)\)}
        \;
    }
\end{algorithm}

\begin{remark}
    Theorem~\ref{lemma:explicit-CRT-mod}
    recovers \(\RNS{x}{\secretprimes}\)
    or \(\RNS{x - \Delta}{\secretprimes}\)
    from \(\RNS{x}{\squirrelsprimes}\).
To guaranteee a result of \(\RNS{x}{\secretprimes}\) 
    requires increasing the precision
    to \(a \sim \log_2\Delta\),
    which means working with integers the size of \(\Delta\);
    but then we may as well reconstruct~\(x\).
\end{remark}

\begin{remark}
    \label{rem:numerator-floor-main}
    As noted in~\cite{BernsteinS07},
    the floors in Line~\ref{alg:ModECRT:numerator}
    of Algorithm~\ref{alg:ModECRT}
    can be computed by 
    repeatedly doubling \(y_j\) modulo \(\squirrelsprime_j\)
    and counting overflows.
\end{remark}

\subsection{Compressed verification for \Squirrels}

Recall that our goal is to verify~\eqref{eq:squirrels-verification-lift}
by evaluating~\eqref{eq:squirrels-verification-reduced},
namely
\[
    \sum\nolimits_{i=1}^{n} c_i\cdot\vchecki
    \equiv
    k\Delta
    \pmod{\secretprime_j}
    \quad
    \text{for each }
    1 \le j \le t
    \,,
\]
but \emph{without} knowing \(k\).
Given a public key \pk
and a compression key 
\(
            \CK
            := 
            \big(
                \vec{r}, 
                (\RNS{\Delta_i}{\vec{r}})_{i=1}^\squirrelsprimecount,
                \RNS{\Delta}{\vec{r}},
                (I_1,\ldots,I_\secretprimecount)
            \big)
\),
for each \(1 \le i \le n\)
we use \ModECRT
to compute
\begin{align*}
    \RNS{\bar v_i}{\secretprimes}
    &
    := 
        \text{\ModECRT}(
            \secretprimes,
            \RNS{\Delta}{\secretprimes},
            (\RNS{\Delta_j}{\secretprimes})_{j=1}^\squirrelsprimecount,
            \RNS{\vchecki}{\squirrelsprimes} = (v_{i,j})_{j=1}^\squirrelsprimecount
        )
    \\
    \label{eq:def-epsilons}
    &
    \phantom{:}= 
    \RNS{\vchecki - \epsilon_i\Delta}{\secretprimes}
    \quad
    \text{where }
    \epsilon_i \in \{0,1\}
    \text{ is unknown}.
\end{align*}
The \(\secretprime_j\) are chosen such that
\(\secretprime_j\nmid\Delta\),
so if we let
\[
    k_j' 
    := 
    \big(\sum\nolimits_{i=1}^{n} c_i\bar v_{i,j}\big)I_j
    \bmod{\secretprime_j}
    \quad
    \text{where}
    \quad
    I_j := \Delta^{-1} \bmod{\secretprime_j}
    \quad
    \text{for } 
    1 \le j \le \secretprimecount
    \,,
\]
then
the system of verification equations~\eqref{eq:squirrels-verification-reduced}
becomes
\begin{equation}
    \label{eq:kprime}
    k_j'
    \equiv
    k'
    \pmod{\secretprime_j}
    \quad
    \text{where}
    \quad
    k' 
    := 
    k + \sum\nolimits_{i=1}^{n-1}\epsilon_ic_i
    \,.
\end{equation}
Lemma~\ref{lemma:k-size} below
shows
that
\begin{itemize}
    \item
        If~\eqref{eq:squirrels-verification-lift} holds,
        then \(k'\)
        is a small integer
        (Table~\ref{tab:k-sizes} gives bounds on \(k'\) 
        for each \Squirrels instance),
        small enough that \(k_j' = k'\) \emph{as an integer}
        for each \(j\).
    \item
        If~\eqref{eq:squirrels-verification-lift} does not hold,
        then
        \(k'\) does not exist,
        and the \(k_j'\) look like random (and generally large) values
        modulo each of the \(\secretprime_j\).
\end{itemize}

This distinction is the basis of our \CVerify for \Squirrels:
we compute \((k_1',\ldots,k_\secretprimecount')\),
and \Accept if \(k_1' = \cdots = k_\secretprimecount'\)
and \(k_1'\) is within the bounds on \(k'\);
otherwise, we \Reject.
Note that even if the adversary has full control over the \(c_i\),
they cannot control the \(k_j'\),
because \secretprimes
and the \(I_j\) are unknown.

\begin{lemma}
    \label{lemma:k-size}
    With 
    \Squirrels parameters
    \(q \ll \Delta\) and \(\lfloor\beta^2\rfloor \ll \Delta\):
    if \(\Vert\vec{s}\Vert_2^2 \le \lfloor\beta^2\rfloor\),
    then
    the integer \(k'\) of~\eqref{eq:kprime}
    satisfies
    \[
        \kmin'
        \le
        k'
        \le
        \kmax'
        \quad
        \text{where }
        \begin{cases}
            \kmin'
            :=
            -\big\lfloor2\sqrt{n\lfloor\beta^2\rfloor}\big\rfloor
            - 1
            \,,
            \\
            \kmax'
            :=
            2(n-1)(q-1) + \big\lfloor2\sqrt{n\lfloor\beta^2\rfloor}\big\rfloor
            + 1
            \,.
        \end{cases}
    \]
\end{lemma}
\begin{proof}
By definition,
\begin{align*}
k' = \left( \sum_{i=1}^{n-1} c_i \beta_i \right) - \frac{c_n}{\Delta} \quad
\text{with} \quad \beta_i = \frac{\vchecki}{\Delta} + \epsilon_i
\end{align*}
for \( 1 \leq i < n \), so
\[
k' + E = H + S
\]
where
\[
H = \sum_{i=1}^{n-1} h_i \beta_i, \quad
S = \sum_{i=1}^{n-1} s_i \beta_i, \quad
E = \frac{c_n}{\Delta}.
\]

For \( 1 \leq i < n \), we have \( 0 \leq \vchecki < \Delta \), so \( 0 \leq \beta_i < 2 \).

Thus,
\[
-S' - E < k' < H' + S' + E
\]
where
\[
H' = 2 \sum_{i=1}^{n-1} h_i
\quad \text{and} \quad
S' = 2 \sum_{i=1}^{n} |s_i|
\]
(note: we include \( |s_n| \) in \( S' \)).

But \( 0 \leq |E| \leq \frac{h_n + |s_n|}{\Delta} < 1 \) because
\[
h_n < q \ll \Delta \quad \text{and} \quad |s_n| \leq \lfloor \beta^2 \rfloor \ll \Delta,
\]
so
\[
-S' - 1 < k' < H' + S' + 1.
\]

Now,
\[
0 \leq S' \leq 2 \sqrt{n} \Vert \vec{s} \Vert_2 \leq 2 \sqrt{n \lfloor \beta^2 \rfloor}, \quad 0 \leq H' \leq 2(n-1)(q-1)
\]

and \( k' \) is an integer, so
\[
\left\lceil -2 \sqrt{n \lfloor \beta^2 \rfloor} \right\rceil - 1 \leq k' \leq 2(n-1)(q-1) + \left\lfloor 2 \sqrt{n \lfloor \beta^2 \rfloor} \right\rfloor + 1,
\]
and the result follows. \qed

\end{proof}

\begin{table}
    \caption{Values of $\kmin'$ and $\kmax'$
        from Lemma~\ref{lemma:k-size}
        for the \Squirrels instances in~\cite{squirrels_sub}.
        Note that
        $\kmax' - \kmin'$ is a 24-bit integer
        except for
        \Squirrels[-V],
        where it is 25 bits.
    }
    \label{tab:k-sizes}
    \centering
    \begin{tabular}{r|rrrrr}
        \toprule
        \emph{Instance}
        & \Squirrels[-I]
        & \Squirrels[-II]
        & \Squirrels[-III]
        & \Squirrels[-IV]
        & \Squirrels[-V]
        \\
        \midrule
        \(\kmin'\) & -91554 & -106640 & -144446 &   -15879 &  -210152 \\
        \(\kmax'\) & 8551824 & 9631610 & 9603896 & 14220809 & 17040602 \\
\bottomrule
    \end{tabular}
\end{table}

The verification key \VK is 
\(
    \big(\secretprimes,
    (I_1,\ldots,I_\secretprimecount),
    (\RNS{\bar v_i}{\secretprimes})_{i=1}^{n-1}\big)
\),
so
\[
    |\VK| = 4(n + 1) \secretprimecount \text{ bytes}
    \quad
    \text{and}
    \quad
    |\PK| = 4(n - 1) \squirrelsprimecount \text{ bytes}
    \,.
\]
The key-compression ratio 
is
\(|\PK|:|\VK|\approx \squirrelsprimecount:\secretprimecount\).
See Table~\ref{tab:sizes_squirrels}
for sample values.

The compression key \CK
requires \(4(\squirrelsprimecount+3)\secretprimecount\) bytes.
The values of \((\RNS{\Delta_i}{\vec{r}})_{i=1}^\squirrelsprimecount\),
\(\RNS{\Delta}{\vec{r}}\),
and \((I_1,\ldots,I_\secretprimecount)\)
may be left out of \CK,
thus reducing \(|\CK|\) to \(4\squirrelsprimecount\) bytes,
but this 
implies recomputing them from \(\vec{r}\)
and \squirrelsprimes
for every verification key generation.

\subsection{The algorithms}

Algorithms~\ref{alg:CKeyGen-Squirrels},
\ref{alg:VKeyGen-Squirrels},
and~\ref{alg:CVerify-Squirrels}
formalise \CKeyGen,
\VKeyGen,
and \CVerify for \Squirrels.

\begin{algorithm}[H]
 \scriptsize
    \caption{\texttt{CKeyGen} (compression key generation) for \squirrels.}
    \label{alg:CKeyGen-Squirrels}
    \Parameters{$n$, $s$, $\Delta$, $\squirrelsprimes$,
        $(\Delta_1,\ldots,\Delta_{\squirrelsprimecount})$,
        $\RNS{\Delta_i}{\squirrelsprimes}$,
        $t$
    }
\KwOut{\CK}
Sample a list \(\vec{r} = (r_1,\ldots,r_t)\) of random 31-bit primes
        \tcp*{Use Lemma~\ref{lemma:31-bit-primegen}}
\label{alg:CKeyGen-Squirrels:setup}
        Compute 
        \((\RNS{\Delta_i}{\vec{r}})_{i=1}^s\) 
        and \(\RNS{\Delta}{\vec{r}}\)
        from \squirrelsprimes
        \tcp*{Using e.g.~Algorithm~\ref{alg:ModECRTSetup}
        in App.~\ref{app:rns}.}
\For{\(1 \le j \le \secretprimecount\)}{
            \(I_j \gets \Delta^{-1} \bmod{\secretprime_j}\)
            \tcp*{Use \(I_j = (\Delta \bmod
            \secretprime_j)^{-1}\bmod{\secretprime_j}\)}
        }
\Return{\(
            \CK
            := 
            \big(
                \vec{r}, 
                (\RNS{\Delta_i}{\vec{r}})_{i=1}^\squirrelsprimecount,
                \RNS{\Delta}{\vec{r}},
                (I_1,\ldots,I_\secretprimecount)
            \big)
        \)}
\end{algorithm}

Algorithm~\ref{alg:CKeyGen-Squirrels}
requires randomly sampling 31-bit primes,
which is easy using the criteria
of~\cite{1980/Pomerance--Selfridge--Wagstaff}.
Recall that an odd integer \(r = d2^u + 1\)
is a \emph{strong pseudoprime to the base \(a\)}
if \(a^d \equiv 1 \pmod{r}\)
or \(a^{d 2^v} \equiv -1 \pmod{r}\)
for some \(0 \le v < u\).

\begin{lemma}
    \label{lemma:31-bit-primegen}
    Let \(2^{30} < r < 2^{31}\) be odd.
    If \(r\) is a strong pseudoprime to the bases
    \(2\), \(3\), and \(5\),
    then either \(r\) is prime
    or
    \(r = 1157839381 = 24061\cdot48121\).
\end{lemma}
\begin{proof}
    See~\cite[Page 1022]{1980/Pomerance--Selfridge--Wagstaff}.
    \qed
\end{proof}

\begin{algorithm}[H]
\scriptsize
    \caption{\texttt{VKeyGen} (verification key generation) for \squirrels.}
    \label{alg:VKeyGen-Squirrels}
    \Parameters{$n$, $s$, $\Delta$, $\squirrelsprimes$,
        $(\Delta_1,\ldots,\Delta_{\squirrelsprimecount})$,
        $\RNS{\Delta_i}{\squirrelsprimes}$,
        $t$
    }
    \KwIn{$\CK =
            \big(
                \vec{r}, 
                (\RNS{\Delta_i}{\vec{r}})_{i=1}^\squirrelsprimecount,
                \RNS{\Delta}{\vec{r}},
                (I_1,\ldots,I_\secretprimecount)
            \big)$,
        $\PK = (\vchecki)_{i=1}^{n-1}$
    }
    \KwOut{\VK}
\ForEach(\tcp*[f]{May be parallelized}){\(1 \le i < n\)}{
            $ \RNS{\bar v_i}{\secretprimes} \gets $
            \ModECRT{\(\secretprimes\),
                \(\RNS{\Delta}{\secretprimes}\),
                \((\RNS{\Delta_i}{\secretprimes})_{i=1}^{\secretprimecount}\),
                \(\vchecki\)
            }
            \tcp*{=\(\RNS{\vchecki +
            \epsilon_i\Delta}{\secretprimes}\)}
        }
\Return{\(
            \big(
                \vec{r},
                (I_1,\ldots,I_\secretprimecount),
                (\bar v_{i,1},\ldots,\bar
                v_{i,\secretprimecount})_{i=1}^{n-1}
            \big)
        \)}
\end{algorithm}

Algorithm~\ref{alg:CVerify-Squirrels}
defines \CVerify.
Lines~\ref{alg:CVerify-Squirrels:start-CT}-\ref{alg:CVerify-Squirrels:end-CT}
must be implemented with constant-time techniques
to ensure no information on \secretprimes is leaked
in the event of rejection\footnote{The implementation can be made 
constant‐time by using constant‐time arithmetic routines 
(e.g. multiplication/modular reduction~\cite{modular_multiplication}) and replacing any conditional
 branches with bit‐masking operations.}.

\begin{algorithm}
\scriptsize
    \caption{\texttt{CVerify} (compressed verification) for \squirrels.}
    \label{alg:CVerify-Squirrels}
    \Parameters{$n$, $s$, $\Delta$, $\squirrelsprimes$,
        $(\Delta_1,\ldots,\Delta_{\squirrelsprimecount})$,
        $\RNS{\Delta_i}{\squirrelsprimes}$,
        $t$
    }
    \KwIn{$\sigma = (\salt,\underline{\vec{s}})$,
          $m$,
          $\VK =
            \big(
                \vec{r},
                (I_1,\ldots,I_\secretprimecount),
                (\bar v_{i,1},\ldots,\bar v_{i,\secretprimecount})_{i=1}^{n}
            \big)$
    }
    \KwOut{\Accept or \Reject}
$\vec{s} \gets \decomp(\underline{\vec{s}})$
        \;
\If{$\vec{s} = \perp$ \textbf{or} $\Vert\vec{s}\Vert_2^2 > \lfloor \beta^2 \rfloor$}{
            \Return \Reject
        }
$\vec{c} \gets \vec{s} + \htopoint(m\parallel\salt, q, n)$
        \;
\((a_1,\ldots,a_\secretprimecount) \gets (\True,\ldots,\True)\)
        \;
\ForEach(\tcp*[f]{May be parallelized}){\(1 \le j \le \secretprimecount\)}{
            \label{alg:CVerify-Squirrels:start-CT}
            \(
                k_j'
                \gets
                \big(\big(\sum_{i=1}^{n}c_i\bar v_{i,j}\big)I_j
                - \kmin'\big)
                \bmod{\secretprime_j}
            \)
            \tcp*[f]{\(\kmin'\): see Lemma~\ref{lemma:k-size}}

            \If(\tcp*[f]{\(\kmax' - \kmin'\): see Lemma~\ref{lemma:k-size}}){
                \(k_j' > \kmax' - \kmin'\)
            }{
                \(a_j \gets \False\)
            }
        }
\uIf{\(
                (a_1\wedge\cdots\wedge a_\secretprimecount)
                \wedge
                (k_1' = \cdots = k_\secretprimecount')
            \)}{
            \Return{\Accept}
        }
        \Else{
            \Return{\Reject}
        }
        \label{alg:CVerify-Squirrels:end-CT}
\end{algorithm}

\subsection{Security argument}
\label{sec:chip-security}

If \(\Sigma\) is \Squirrels,
then \(\Ring = \Module = \ZZ\).
Compressed \Squirrels
samples \Kernel from
\[
    \KernelSet 
    = 
    \{
        \Delta'\ZZ 
        : 
        \Delta' \text{ is a product of \secretprimecount 31-bit primes}
    \}
    \,.
\]
Theorem~\ref{thm:unforgeability}
ensures EUF-CMA security for compressed \Squirrels
provided \(\SGP(\KernelSet,\TSet(\Sigma))\) is hard.
For \Squirrels parameters,
\(\TSet(\Sigma)\subset [0,2q\sqrt{n}\lfloor{\beta^2}\Delta\rfloor] \subset [0,2^{30}\Delta)\):
at most \(\squirrelsprimecount\) 31-bit primes can divide any \(\vec{t}
\in \TSet(\Sigma)\),
so
\(
    \kappa_{\TSet}
    \approx
    \squirrelsprimecount!/(\secretprimecount!(\squirrelsprimecount-\secretprimecount)!)
\).
Now \(\#\KernelSet =
P_{31}!/(\secretprimecount!(P_{31}-\secretprimecount)!)\)
where \(P_{31} := \#\Primes{31}\),
and \(\#(\Module/\Kernel) \approx 2^{31\secretprimecount}\).
When \(\secretprimecount \ll \squirrelsprimecount\),
we have
\(\squirrelsprimecount!/(\squirrelsprimecount-\secretprimecount)! \sim
\squirrelsprimecount^\secretprimecount\) and
\(P_{31}!/(P_{31}-\secretprimecount)! \sim P_{31}^\secretprimecount\).
Looking at~\cite[Table 3]{Wagstaff06},
we find that
\[
    P_{31} := \#\Primes{31} = 105097565-54400028 \approx 2^{25.6}\,.
\]
The heuristic of~\S\ref{sec:general-hardness}
therefore suggests
taking \(t \approx \mu/(25.6-\log_2(s))\), where $\mu$ is the targeted security level.

In reality,
it is computationally infeasible to construct forgery attempts
\((m,\sigma)\)
that yield \(\vec{t}\) with \(\#\KernelSet_{\vec{t}} \approx \kappa_{\TSet}\):
this would mean constructing \((m,\sigma)\)
such that \(\vec{t}\) is divisible by another product \(\Delta'\)
of \squirrelsprimecount 31-bit primes,
and this is essentially as hard as forging a full \Squirrels signature:
that is, solving \(\mathrm{GSIS}_{n,\Delta,\beta}\).
If we want an \((m,\sigma)\) mapping to a \(\vec{t}\) with \(k\) 31-bit
prime factors, and thus \(\#\KernelSet_{\vec{t}} = \binom{k}{t}/\#S\),
we must solve a single GSIS instance with modulus \(\Delta^{\star}\ll\Delta\).  
Indeed, it requires a nontrivial computational effort
to even construct a forgery attempt \((m,\sigma)\) yielding \(\vec{t}\)
with \(\#\KernelSet_{\vec{t}} > 0\).
In our security estimates, we therefore model the expected value of
\(\#\KernelSet_{\vec{t}}\) as a small constant,
and hence
we choose \secretprimecount such that
$P_{31}!/(\secretprimecount!(P_{31}-\secretprimecount)!$
is on the order of $2^\lambda$,
which suggests taking \secretprimecount as in Table~\ref{tab:secretprimecount}. 
While the resulting security level \(\mu\) is slightly smaller
than \(\lambda\) in some cases,
it should be remembered that 
each forgery attempt requires
an interaction with the verifier, and is therefore substantially more
expensive than (e.g.) the AES circuit evaluations used to model
post-quantum cryptographic attack costs.

\begin{table}[htp]
    \caption{Suggested values for the number \secretprimecount of secret
        primes in \secretprimes.
    }
    \label{tab:secretprimecount}
    \label{tab:sizes_squirrels}
    \centering
    \begin{tabular}{l|r@{\;\;}r@{\;\;}r@{\;\;}r|r@{\;\;}rr@{\;\;}rr}
        \toprule
        & \multicolumn{4}{c|}{Original scheme}
        & \multicolumn{5}{c}{Compressed verification}
        \\
        Instance
        & \(\lambda\) & \squirrelsprimecount & \(|\PK|\) & \(|\sigma|\)
        & \secretprimecount 
        & \(\mu\) 
        & \(|\CK|\) & \(|\VK|\) & \(|\PK|:|\VK|\)
        \\
        \midrule
\Squirrels[-I]
        & 128 & 165 &    681780 & 1019
& 5 & 121.1 & 3360 & 20700 & 32.94
        \\
\Squirrels[-II]
        & 128 & 188 &    874576 & 1147
& 5 & 121.1 & 3820 & 23300 & 37.54
\\
\Squirrels[-III]
        & 192 & 262 & 1629640   & 1554
& 8 & 189.5 & 8480 & 49824 & 32.71
\\
\Squirrels[-IV]
        & 192 & 275 & 1888700   & 1676
& 8 & 189.5 & 8896 & 55008 & 34.34
        \\
\Squirrels[-V]
        & 256 & 339 & 2786580   & 2025
        & 11 & 256.3 & 15048 & 90508 & 30.79
        \\
\bottomrule
    \end{tabular}
\end{table}
\vspace{-0.5cm}

\section{Compressed verification for \Wave
}\label{sec:sketch_wave_rip}

\Wave is a GPV-style signature
based on hard problems in ternary linear codes.
Very briefly: a \Wave public key is a matrix 
\(\PK = \vec{R} \in \F_3^{k\times(n-k)}\)
such that
\(\vec{M} = (\vec{I}_{n-k}|\vec{R})^\top\) in \(\F_3^{n\times(n-k)}\)
is a parity-check matrix for a permuted generalized \((U|U+V)\)-code \(C\);
knowledge of the relation between \(C\) and the component codes
\(U\) and \(V\) is a trapdoor allowing the signer
to generate vectors \(\vec{s}\) in \(\F_3^n\)
such that
\begin{equation}
    \label{eq:wave-orig-verif}
    \Constraint{\vec{s}}
    \qquad
    \text{and}
    \qquad
    \vec{s}\vec{M} = \hash(\salt,m)
    \,,
\end{equation}
where \(\Constraint{\vec{s}}\) is that
\(\vec{s}\) have a fixed, high weight \(w\).
An ``original'' \Wave signature (as in~\cite{Wave})
is \(\sigma = (\salt,\vec{s})\).
In the \Wave NIST submission~\cite{Wave-submission},
\(\vec{s}\) is truncated to its last \(k\) entries
(the other \(n-k\) entries are implicitly recovered in verification).
The special form of \(\vec{M} = (\vec{I}_{n-k}|\vec{R})^\top\)
allows us to rewrite~\eqref{eq:wave-orig-verif}
as
\begin{equation}
    \label{eq:wave-alt-verif}
    \Constraint{\vec{s}}
    \quad
    \text{and}
    \quad
    \vec{t}\vec{M} 
    =
    \vec{0}
    \quad
    \text{where}
    \quad
    \vec{t} := \vec{s} - (\hash(\salt,\vec{s})|\vec{0}_k)
    \,.
\end{equation}

\Wave has exceptionally large public keys: for example,  
\Wave[822], \Wave[1249], and \Wave[1644] require  
\SI{3.5}{\mega\byte}, \SI{7.5}{\mega\byte}, and  
\SI{13}{\mega\byte}, respectively.  
This motivates the use of \emph{compressed verification},  
where the verifier privately replaces the large matrix  
\(\vec{M}\) with a much smaller secret key \(\VK\) derived  
via a compression matrix.

\VKeyGen computes \(\VK :=  
(\mathbf{I}_{n-k} \mid \mathbf{R})^\top \vec{C}\),  
where \(\vec{C}\) is the verifier’s compression matrix.  
Only the bottom \(n-k-c\) rows of \(\VK\) need to be stored,  
yielding a total size of \(c(n-c)/4\) bytes. \CVerify then 
replaces the full-rank check  
\(\vec{t}\vec{M} = \vec{0}\) with a lower-dimensional test  
\(\vec{t}\VK = \vec{0}\) over \(\FF_3^c\),  
as detailed in Algorithm~\ref{alg:vkverify_ripple}.
Compressed verification requires ``original'' \Wave signatures 
(as in~\cite{Wave}) and is incompatible with the truncated 
versions from~\cite{Wave-submission}. Although full signatures 
can be recovered from truncated ones using the public key, 
doing so during verification defeats compression by reintroducing 
computational and storage overhead.

\begin{algorithm}[htp]
    \caption{\texttt{CVerify} (compressed verification) for \Wave}\label{alg:vkverify_ripple}
    \KwIn{message $m$, signature $\sigma = (\salt,\vec{s})$, and verification key $\vk$}
    \KwOut{\Accept or \Reject}
\If{$\weight(\vec{s}) \neq W$}{
        \KwRet{\Reject}
    }
    $\vec{t} \gets \vec{s} - (\hash( \salt \parallel m )\parallel \vec{0}_k)$
    \;
    $\vec{r} \gets \vec{t}\vk$
    \;
    \If{$\vec{r} \neq 0$}{
        \KwRet{\Reject}
    }
    \KwRet{\Accept}
    \;
\end{algorithm}

The security argument for compressed \Wave
is particularly simple.
Theorem~\ref{thm:unforgeability} ensures that 
\(\Sigmacomp[\Kernel]\) is EUF-CMA secure if
\(\SGP(\KernelSet,\TSet(\Sigma))\)
is hard,
where \KernelSet is the set of codimension-$c$
subspaces of $\Module = \FF_3^{n-k}$
and \(\TSet(\Sigma) = \Module\).
We have \(\#\KernelSet = \threebinom{n-k}{n-k-c}\),
where \threebinom{\cdot}{\cdot} is the $3$-binomial coefficient,
\(\#(\Module/\Kernel) = 3^{n-k-c}\),
and \(\#\KernelSet_{\vec{t}} = \kappa_{\TSet(\Sigma)} = \threebinom{n-k-1}{n-k-c-1}\)
for \emph{every} \(\vec{t} \not= \vec{0} \in \Module\).
This tells us that
a naive adversary that simply tries
random forgery attempts \((m,\sigma)\) is already close to optimal.
Intuitively:
since \(\vec{C}\) is a random projection,
if \(\vec{t}\not=\vec{0}\) then \(\vec{t}\vec{C}\) is a random element of \(\FF_3^c\),
which is \(\vec{0}\) (leading to an \Accept) with probability \(1/3^c\).

Indeed, if we admit the heuristic of~\S\ref{sec:general-hardness}
then we can take \(2^\mu \approx 3^c\):
that is, \(c \approx \log_3(2)\mu\).
Taking \(c\) to be a multiple of 8 simplifies implementation.
Table~\ref{tab:Wave-parameters} lists suggested values of \(c\)  
and corresponding sizes of \CK and \VK
for \Wave[822], \Wave[1249], and \Wave[1644].

\begin{table}
    \caption{Parameters for (compressed) \Wave.
        Sizes are in bytes.
        \Wave signatures are variable-length:
        the values of \(|\sigma|\) here are upper bounds,
        and \(|\sigma|\) roughly doubles for the ``original''
        (non-truncated) signatures required for compressed verification.
    }
    \label{tab:Wave-parameters}
    \centering
    \resizebox{\textwidth}{!}{
        \begin{tabular}{l|cc|rr|r@{\;\;}r@{\;\;}r@{\;\;}rr}
            \toprule
            & \multicolumn{2}{c|}{Security level} 
            & \multicolumn{2}{c|}{Original scheme}
            & \multicolumn{5}{c}{Compressed verification}
            \\
            Instance & 
            NIST PQ & \(\lambda\) 
            & \(|\sigma|\) & \(|\PK|\)
            & \(c\) & \(\mu\) & \(|\CK|\) & \(|\VK|\) & \(|\PK|:|\VK|\)
            \\
            \midrule
            \Wave[822]
            & 1 & 128
            & 822 &  {3\,677\,390}
& 80 & 126.8 & {83\,822} & {171\,594} & {21.4}
            \\
            \Wave[1249]
            & 3 & 192
            & 1249 &  {7\,867\,598}
& 120 & 190.2 & {183\,134} & {266\,188} & {29.6}
            \\
            \Wave[1644]
            & 5 & 256
            & 1644 & {13\,632\,308}
            & 160 & 253.6 & {321\,507} & {370\,436} & {36.8}
\\
            \bottomrule
        \end{tabular}
    }
\end{table}

\bibliographystyle{abbrv}
\bibliography{references}

\appendix

\section{Subroutines for the explicit CRT
}\label{app:rns}

\label{sec:ECRT}

We maintain the notation of~\S\ref{sec:sq_mar}:
\(
    \squirrelsprimes 
    =
    (\squirrelsprime_1,\ldots,\squirrelsprime_\squirrelsprimecount)
\) 
is a list of distinct primes,
\(\Delta := \prod_{i=1}^{\squirrelsprimecount}\squirrelsprime_i\)
is their product,
\(\Delta_i := \Delta/\squirrelsprime_i\)
and \(q_i := \Delta_i^{-1}\pmod{\squirrelsprime_i}\)
for \(1 \le i \le \squirrelsprimecount\).
Algorithm~\ref{alg:qCoefficients}
computes \((q_1,\ldots,q_s)\).
Algorithm~\ref{alg:ModECRTSetup}
computes
\((\RNS{\Delta_i}{\secretprimes})_{i=1}^s\) 
and \(\RNS{\Delta}{\secretprimes}\).
given another list of primes
\(
    \secretprimes 
    = 
    (\secretprime_1,\ldots,\secretprime_\secretprimecount)
\)
(all prime to \(\Delta\)).
These algorithms are not optimal,
but they avoid multiprecision arithmetic.

\begin{algorithm}[ht]
    \caption{Explicit Modular CRT setup: \(q\)-coefficients.}
    \label{alg:qCoefficients}
    \KwIn{\(\squirrelsprimes = (\squirrelsprime_1,\ldots,\squirrelsprime_s)\)}
    \KwOut{\((q_1,\ldots,q_s)\)
        s.t. \(0 < q_i < \squirrelsprime_i\)
        and \(q_i(\Delta/\squirrelsprime_i) \equiv 1 \pmod{\squirrelsprime_i}\)
        for \(1 \le i \le s\)
    }
    \Function{\qCoefficients{\((\squirrelsprime_1,\ldots,\squirrelsprime_s)\)}}{
        \((q_1,\ldots,q_s) \gets (1,\ldots,1)\)
        \;
\ForEach{\(1 \le i \le s\)}{
            \ForEach{\(1 \le j \le s\), \(j \not= i\)}{
                \(q_i \gets q_i\cdot \squirrelsprime_j\bmod{\squirrelsprime_i}\)
            }
            \(q_i \gets q_i^{-1} \bmod{\squirrelsprime_i}\)
            \;
        }
        \Return{\((q_1,\ldots,q_s)\)}
    }
\end{algorithm}

\begin{algorithm}[ht]
    \caption{Explicit Modular CRT setup.}
    \label{alg:ModECRTSetup}
\KwIn{
        \(\squirrelsprimes = (\squirrelsprime_i)_{i=1}^\squirrelsprimecount\)
        and 
        \(\secretprimes = (\secretprime_j)_{j=1}^\secretprimecount\).
    }
    \KwOut{
        \((\RNS{\Delta_i}{\secretprimes})_{i=1}^\squirrelsprimecount\)
        and \(\RNS{\Delta}{\secretprimes}\).
    }
    \Function{\ModECRTSetup{\squirrelsprimes, \secretprimes}}{
        \(\vec{m} \gets \RNS{1}{\secretprimes}\)
        \;
\(
            (\vec{c}^{(1)},\ldots,\vec{c}^{(\squirrelsprimecount)})
            \gets
            (\RNS{1}{\secretprimes},\ldots,\RNS{1}{\secretprimes})
        \)
        \;
\ForEach{\(1 \le i \le \squirrelsprimecount\)}{
            \(
                \vec{u} \gets \RNS{p_i}{\secretprimes}
            \)
            \label{alg:ModECRTSetup:pimodrj}
            \tcp*{\(u_j = p_j\) or \(p_j-r_j\)}
\(
                \vec{m}
                \gets
                (
                    m_1u_1\bmod{\secretprime_1},
                    \ldots,
                    m_\secretprimecount
                    u_\secretprimecount\bmod{\secretprime_\secretprimecount}
                )
            \)
            \;
            \ForEach{\(1 \le j < i\) \textbf{and} \(i < j \le \squirrelsprimecount\)}{
                \(
                    \vec{c}^{(j)}
                    \gets
                    (
                        c^{(j)}_1 u_1\bmod{\secretprime_1},
                        \ldots,
                        c^{(j)}_\secretprimecount u_\secretprimecount\bmod{\secretprime_\secretprimecount}
                    )
                \)
                \;
            }
        }
        \Return{\((\vec{c}^{(1)},\ldots,\vec{c}^{(\squirrelsprimecount)})\), \(\vec{m}\)}
        \;
    }
\end{algorithm}

\end{document}